\documentclass[runningheads]{llncs}
\usepackage{booktabs}   %% For formal tables:
\usepackage{subcaption} %% For complex figures with subfigures/subcaptions
\usepackage[T1]{fontenc} % needed for scaling fancy fonts (?)
\usepackage[utf8]{inputenc} % not sure this is needed
\usepackage{graphicx}
\usepackage{amssymb}
\usepackage{multirow}

\usepackage[scaled=.92]{helvet}
\usepackage{times}

\usepackage[final]{listings}
\usepackage{silver}
\usepackage{tikz}
\usetikzlibrary{matrix,arrows,positioning,calc,fit,backgrounds}
\usepackage{wrapfig}

\usepackage[shortlabels, inline]{enumitem}
\makeatletter
\newcommand{\myitem}[1][]{
  \protected@edef\@currentlabel{#1}%
\item[#1]
}
\makeatother

\usepackage{mathtools}

\usepackage{etoolbox}

\usepackage{array}
\newcolumntype{L}{>$l<$}

\usepackage{stmaryrd}

\usepackage{dashbox}

\usepackage{scalerel}

\usepackage{mathpartir}
\usepackage{pftools}  % A local package
\usepackage{xspace}

\usepackage{cite}
\newcommand{\citet}[1]{\textcolor{red}{\cite{#1}}}

\usepackage[pdftex,pdfpagelabels,bookmarks,hyperindex,hyperfigures]{hyperref}
\newcommand{\nocomments}{}

\newcommand\nocolour{}

\newcommand\notodos{}

\usepackage[normalem]{ulem}

\newcommand{\as}[1]  {\ifdefined\nocolour{#1}\else{\color{purple}{#1}}\fi}

\newcommand{\asfootnote}[1]{\ifdefined\nocomments{}\else\as{\footnote{\as{ALEX: #1}}}\fi}
\newcommand{\tw}[1]  {\ifdefined\nocolour{#1}\else{\color{blue}{#1}}\fi}
\newcommand{\twout}[1]{\ifdefined\nocolour{}\else\tw{{\sout{#1}}}\fi}
\newcommand{\twfootnote}[1]{\ifdefined\nocomments{}\else\tw{\footnote{\tw{THOMAS: #1}}}\fi}
\newcommand{\sk}[1]  {\ifdefined\nocolour{#1}\else{\color{green!50!black}{#1}}\fi}

\newcommand{\skfootnote}[1]{\ifdefined\nocomments{}\else\sk{\footnote{\sk{SID: #1}}}\fi}

\newcommand{\fcomment}[2]{\ifdefined\nocomments{}\else\footnote{\textcolor{red}{#1:} #2}\fi}

\newcommand{\ftodo}[1]{\ifdefined\notodos{}\else\fcomment{TODO}{#1}\fi}

\usepackage{iris}

\newcommand{\moreless}[2]{#1}
\newcommand{\techreport}[1]{\moreless{#1}{}}

\newcommand{\ifspace}[1]{}

\newcommand{\defFunc}[2]{\newcommand{#1}{\mathsf{#2}}}

\newcommand{\Nat}{\mathbb{N}}

\newcommand{\impl}{\Rightarrow}
\renewcommand{\ra}{\rightarrow}

\renewcommand{\proves}{\vdash}

\renewcommand{\dom}{\operatorname{\mathsf{dom}}}

\newcommand{\pto}{\rightharpoonup}

\newcommand{\defeq}{\coloneqq}

\newcommand{\pipe}{\triangleright}

\newcommand{\paren} [1] {\ensuremath{ \left( {#1} \right) }}

\newcommand{\abs}[1]{\ensuremath{\lvert #1 \rvert}}
\newcommand{\setcomp}[2]{\ensuremath{\left\{#1\;\middle|\;#2\right\}}}

\newcommand{\rSc}[1]{\S\ref{#1}}
\newcommand{\rF}[1]{Figure~\ref{#1}}
\newcommand{\rD}[1]{Definition~\ref{#1}}
\newcommand{\rL}[1]{Lemma~\ref{#1}}

\newcommand{\rE}[1]{Example~\ref{#1}}

\newcommand{\rEq}[1]{\ensuremath{(\ref{#1})}}
\newcommand{\refRule}[1]{(\ref{#1})}
\newcommand{\refApp}[1]{Appendix~\ref{#1}}

\newcommand{\code}[1]{\textnormal{\small\texttt{#1}}}
\newcommand{\eg}{e.g.\@}
\newcommand{\ie}{i.e.\@}
\newcommand{\cf}{cf.\@}
\newcommand{\etal}{et al.\@}

\newcommand{\entails}{\ensuremath{\models}}

\newcommand*{\figref}[1]{Figure~\ref{fig-#1}}
\newcommand*{\secref}[1]{\S\ref{sec-#1}}

\newcommand*{\defref}[1]{Definition~\ref{def-#1}}

\makeatletter
\def\operator#1{\@ifnextchar\bgroup {\operatorarg{\ensuremath{#1}}}{\ensuremath{#1}}}
\def\operatorarg#1#2{{#1}{\ensuremath{(#2)}}}
\def\spoperator#1#2{\@ifnextchar\bgroup{\spoperatorarg{\ensuremath{#1}}{\ensuremath{#2}}}{\ensuremath{#1}}}
\def\spoperatorarg#1#2#3{\ensuremath{#1#2#3}}
\makeatother

\newskip \point \point =1pt
\setbox134=\hbox{\leavevmode\raise0\point\hbox{${\langle}\kern-2.5\point{\langle}$}}
\setbox135=\hbox{\raise0\point\hbox{${ \rangle}\kern-2.5\point{ \rangle}$}}

\setbox136=\hbox{\leavevmode\raise0\point\hbox{${\lfloor}\kern-3.0\point{\lfloor}$}}
\setbox137=\hbox{\raise0\point\hbox{${ \rfloor}\kern-3.0\point{ \rfloor}$}}

\setbox138=\hbox{\leavevmode\raise0\point\hbox{${[}\kern-1.5\point{[}$}}
\setbox139=\hbox{\raise0\point\hbox{${]}\kern-1.5\point{]}$}}

\newcommand{\framework}{foundational flow framework\xspace}

\newcommand{\FRAMEWORK}{Foundational Flow Framework\xspace}

\newcommand{\emptySubscript}{\emptyset}

\newcommand{\field}[1]{\text{\code{#1}}}
\newcommand{\nextField}{\field{next}}
\newcommand{\fnextField}{\field{fnext}}

\newcommand{\values}{\mathsf{Val}}

\makeatletter%
\@ifundefined{dplus}{%
\newcommand\dplus{\mathbin{+\kern-1.0ex+}}
}{}
\makeatother%

\newcommand{\globalInt}{\varphi}

\defFunc{\inset}{ins}
\defFunc{\linkset}{lnks}
\defFunc{\outset}{outs}
\defFunc{\edgeset}{es}
\defFunc{\reachset}{rs}
\defFunc{\contents}{cn}
\defFunc{\inreach}{inr}

\newcommand{\graph}{G}

\newcommand{\nodeDom}{\mathfrak{N}}%{\mathsf{Node}}

\newcommand{\mDom}{M}
\newcommand{\mOp}{+}
\newcommand{\mPlus}{\mOp} % just an alias
\newcommand{\mplus}{\mPlus} % just an alias
\newcommand{\mBigOp}{\sum}
\newcommand{\mLeq}{\leq}
\newcommand{\mZero}{0}
\newcommand{\mVar}{m}

\newcommand{\edgeDom}{E}

\newcommand{\fGraph}{H}
\newcommand{\fGraphDom}{\mathsf{FG}}
\newcommand{\fGraphComp}{\odot}
\newcommand{\fGraphBigComp}{\bigodot}
\newcommand{\fGraphEmpty}{\fGraph_{\emptySubscript}}

\newcommand{\fGraphMap}{\mathcal{H}}

\newcommand{\edgeFn}{\operator{e}}

\newcommand{\flowVar}{\textit{fl}}
\newcommand{\flow}{\operator{\mathsf{flow}}}%{\mathit{flow}}
\newcommand{\flowEqn}{\mathsf{FlowEqn}}

\newcommand{\edgeFnEmpty}{\edgeFn_{\emptySubscript}}
\newcommand{\flowEmpty}{\flow_{\emptySubscript}}

\newcommand{\morphisms}{\mathsf{End}}
\newcommand{\deltaFn}[2]{\delta_{#1 = #2}}

\newcommand{\fnComp}{\circ}
\newcommand{\zeroFn}{\lambda_{\mZero}}
\newcommand{\identityFn}{\lambda_{\mathsf{id}}}
\newcommand{\lambdaFn}[2]{(\lambda #1.\; #2)}

\newcommand{\capacityExtendedBy}{\precsim_s}

\defFunc{\inflowFn}{inf}
\defFunc{\outflowFn}{outf}

\newcommand{\outflow}{\mathit{out}}

\newcommand{\interfaces}{\mathsf{FI}}
\newcommand{\interface}{I}
\defFunc{\interfaceFn}{int}
\defFunc{\footprintFn}{ffp}
\newcommand{\interfaceEmpty}{\interface_{\emptySubscript}}

\newcommand{\intComp}{\oplus}

\newcommand{\intLessEquiv}{\precsim}

\newcommand{\inflowOfInt}[1]{{#1}.\inflow}

\newcommand{\outflowOfInt}[1]{{#1}.\outflow}

\defFunc{\init}{init}
\defFunc{\capacity}{cap}
\defFunc{\capacityAux}{capAux}
\defFunc{\flowFn}{flow}
\defFunc{\flowmapFn}{flm}
\defFunc{\userEdgeFn}{edges}

\defFunc{\goodCondition}{\gamma}

\defFunc{\pathCount}{pc}

\newcommand{\inflow}{\mathit{in}}

\defFunc{\lock}{lock}
\defFunc{\keyset}{ks}
\defFunc{\pathset}{path}
\defFunc{\inflows}{In}

\defFunc{\composition}{comp}
\defFunc{\projection}{proj}

\newcommand{\vars}{\mathsf{Var}}
\newcommand{\addrs}{\mathsf{Loc}}
\newcommand{\fieldDom}{\mathsf{Field}}
\newcommand{\cmds}{\mathsf{Com}}
\newcommand{\states}{\mathsf{Heap}}
\newcommand{\state}{\sigma}

\newcommand{\stateComp}{\circledcirc}

\newcommand{\heap}{h}

\newcommand{\interp}{i}

\newcommand{\fields}{\mathit{fs}}

\defFunc{\abstractionFn}{edge}
\defFunc{\hrepSpatial}{spatialRep}

\newcommand{\mkblue}[1]{\textcolor{blue}{#1}}
\newcommand{\mkpurple}[1]{\textcolor{purple}{#1}}
\newcommand{\annotOpt}[1]{\mkpurple{\left\{\begin{aligned}#1\end{aligned}\right\}}}
\newcommand{\annot}[1]{\mkblue{\left\{\begin{aligned}#1\end{aligned}\right\}}}

\newcommand{\hoareTriple}[3]{\annot{#1} \; #2 \; \annot{#3}}

\newcommand{\denotation}[1]{\llbracket #1 \rrbracket}

\DeclareMathOperator*{\Sep}{\scalerel*{\ast}{\sum}}
\newcommand{\ite}[3]{\paren{#1 \;?\; #2 : #3}}

\newcommand{\lseg}{\mathsf{lseg}}

\newcommand{\emp}{\mathit{emp}}
\newcommand{\true}{\mathit{true}}

\newcommand{\nullVal}{\mathit{null}}

\newcommand{\graphPred}{\mathsf{Gr}}
\newcommand{\graphPredEA}{\mathsf{Gr_{a}}}

\newcommand{\nodePred}{\mathsf{N}}
\newcommand{\nodePredEA}{\mathsf{N_{a}}}

\newcommand{\emptyFn}{\epsilon}

\newcommand{\goesto}{\rightarrowtail}

\newcommand{\skipCommand}{\text{\texttt{\textbf{skip}}}}
\newcommand{\assumeCommand}{\text{\texttt{\textbf{assume}}}}
\newcommand{\allocCommand}{\text{\texttt{\textbf{alloc}}}}

\newcommand{\reduces}[2]{\xrightarrow[#2]{#1}}

\newcommand{\faultConfig}{\text{\texttt{\textbf{fault}}}}

\newcommand{\freeListHead}{\mathit{fh}}
\newcommand{\freeListTail}{\mathit{ft}}
\newcommand{\mainListHead}{\mathit{mh}}

\newcommand{\marked}{\blacklozenge}
\newcommand{\unmarked}{\diamondsuit}

\lstdefinelanguage{SPL}{
  morekeywords={acc, method, struct,if,else,returns,procedure,requires,ensures,:=,var,
    new,old,free,implicit,modifies,call,locals,assume,assert,choose,havoc,ghost,
    predicate,function,invariant,while, return,atomic, split, type, field, result,
    define, datatype, domain, axiom},
  deletekeywords={union,int},
  numbers=left,
  xleftmargin=2em,
  escapeinside={@}{@},
  numberstyle=\tiny,
  basicstyle=\footnotesize\ttfamily,
  columns=flexible,
  morecomment=*[s][\color{green!60!black}]{/*}{*/},
  morecomment=*[l][\color{green!60!black}]{//},
  moredelim=**[is][\color{purple}]{|<}{>|},
  mathescape=true,
}
\tikzset{%
  array/.style={matrix of nodes,nodes={draw, minimum size=5mm, anchor=center},column sep=-\pgflinewidth, row sep=-\pgflinewidth, nodes in empty cells,anchor=center},
  ptr/.style={*->, shorten <=-(1.8pt+1.4\pgflinewidth)},
  edge/.style={->},
  dedge/.style={<->, dashed},
  fedge/.style={->, dashed},
  unode/.style={circle, draw=black, thick, minimum size=8mm},
  mnode/.style={circle, draw=black, thick, fill=gray!20, minimum size=8mm},
  stackVar/.style={circle, fill=none, inner sep=0pt, minimum size=8mm, font=\normalsize},
  gnode/.style={circle, draw=black, thick, minimum size=8mm},
  pnode/.style={circle, draw=black, thick, minimum size=8mm},
  rnode/.style={draw=black, thick, minimum size=8mm},
  prio/.style={circle, fill=none, inner sep=0pt, minimum size=8mm, font=\footnotesize},
  dnode/.style={circle, draw=black, thick, dotted, minimum size=8mm},
  inflow/.style={circle, fill=none, inner sep=0pt, minimum size=5mm, font=\normalsize},
  phantomNode/.style={circle, fill=none, inner sep=0pt, minimum
    size=0pt}
}

\begin{document}

%% Title information
\title{Local Reasoning for Global Graph Properties}

%% Author information
\author{Siddharth Krishna\inst{1} \and
Alexander J. Summers\inst{2} \and
Thomas Wies\inst{1}}
\authorrunning{S. Krishna et al.}
\institute{New York University, USA, \email{\{siddharth,wies\}@cs.nyu.edu}
  \and
  ETH Z\"urich, Switzerland, \email{alexander.summers@inf.ethz.ch}}
\maketitle              % typeset the header of the contribution

\begin{abstract}
  Separation logics are widely used for verifying
  programs that manipulate complex heap-based data structures. %state involving
  %graphs (e.g. the heap).
  These logics build on so-called
  \emph{separation algebras}, which allow expressing properties of heap
%  graph
regions
  such that modifications to a region
  % become \emph{frame preserving}, i.e.
  do not invalidate properties stated about the remainder of the
  heap. %graph.
  This concept is key to enabling modular reasoning and also extends to concurrency.
  %One can then reason locally about the effects of these graph
  %modifications, which greatly simplifies the verification.
  %
  While heaps are naturally related to mathematical graphs, %However,
   many ubiquitous graph properties are non-local in
  character, such as reachability between nodes, path lengths,
  acyclicity and other structural invariants, as well as data
  invariants which combine with these notions. %\asfootnote{probably a
%    bit long. Is it worth trying to sneak overlays in somewhere?}
  Reasoning modularly about such
  % reasoning about global
  graph properties
%  (e.g. reachability)
remains notoriously difficult, since a local modification can have side-effects on a global property that cannot be easily confined to a small region.

  In this paper, we address the question: What separation algebra can be used to avoid
  proof arguments reverting back to
%  is the separation
%  algebra that avoids the proof argument from reverting back to
  tedious global reasoning in such cases? To this end, we consider a
  general class of global graph properties expressed as fixpoints of
  algebraic equations over graphs. We present mathematical foundations
  for reasoning about this class of properties, imposing minimal
  requirements on the underlying theory
   that allow us to define a
  suitable separation algebra. Building on this theory we develop a
  general proof technique for modular reasoning about global graph
  properties over program heaps, in a way which
   can be integrated with existing
   separation logics. To demonstrate our approach, we present local
   proofs for two challenging examples: a
   priority inheritance protocol and the non-blocking concurrent Harris list.
  % We further devise a strategy for automating this
  % technique using SMT-based verification tools.
  % %
  % We have implemented this strategy on top of the verification tool
  % Viper and applied it successfully to a variety of challenging
  % benchmarks including 1) algorithms involving general graphs such as
  % Dijkstra's algorithm and a priority inheritance protocol, 2)
  % inductive data structures such as linked lists and B trees, 3)
  % overlaid data structures such as the Harris list and threaded trees,
  % and 4) OO design patterns such as Composite and
  % Subject/Observer. We are not aware of any single other
  % approach that can handle these examples with the same degree of
  % simplicity or automation.

  % However, certain data structure idioms prevalent in real-world
  % programs are notoriously difficult to reason about, even in these
  % advanced logics (e.g., random access into inductively defined
  % structures, data structure overlays).
  %

\end{abstract}

% ----------------------------------------
% Actual paper
% ----------------------------------------

\section{Introduction}\label{sec-introduction}
Separation logic (SL)~\cite{DBLP:conf/csl/OHearnRY01,
  DBLP:conf/lics/Reynolds02} provides the basis of many successful
verification tools that can verify programs manipulating complex data
structures~\cite{DBLP:conf/nfm/CalcagnoDDGHLOP15,DBLP:conf/nfm/JacobsSPVPP11,viper,
  DBLP:conf/nfm/Appel12}. This success is due to the logic's support
for reasoning modularly about modifications to heap-based data.
%graph.
 For simple inductive data structures such as lists and trees,
much of this reasoning can be
automated~\cite{DBLP:conf/fsttcs/BerdineCO04,
  %DBLP:conf/atva/IosifRV14,
  DBLP:conf/tacas/KatelaanMZ19,
  DBLP:conf/cav/PiskacWZ13,DBLP:conf/nfm/EneaLSV17}.
However, these techniques often fail when data structures are less
regular (e.g. multiple overlaid data structures) or provide multiple
traversal patterns (e.g. threaded trees). Such idioms are prevalent in
real-world implementations such as the fine-grained concurrent data
structures found in operating systems and databases. Solutions to
these problems have been proposed~\cite{DBLP:conf/popl/HoborV13} but
remain difficult to automate. For proofs of general graph algorithms,
the situation is even more dire. Despite substantial improvements in
the verification methodology for such
algorithms~\cite{DBLP:conf/pldi/SergeyNB15,
  DBLP:conf/aplas/RaadHVG16}, significant parts of the proof argument
still typically need to be carried out using non-local
reasoning~\cite{DBLP:conf/itp/ChenCLMT19,DBLP:journals/jar/ChargueraudP19,DBLP:journals/afp/HaslbeckLB19,DBLP:journals/jar/LammichS19}. This paper presents a general technique for %automated
local reasoning about global graph properties \tw{that can be used within off-the-shelf
  separation logics}.
%and applies to a broad class of data structures and graph
%algorithms.
\tw{We demonstrate our technique using two challenging examples
  for which no fully local proof existed before, respectively, whose proof
  required a tailor-made logic.}
  %that are difficult to automate}.
%\tw{In fact, for many of the
%examples that we consider in this work, no fully-modular proof had
%existed before.}

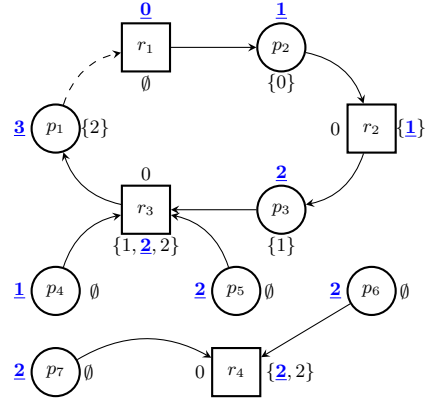
\begin{figure}[t]
  \vspace{-.5em}
  \centering
  \begin{minipage}{.56\linewidth}
  \begin{lstlisting}[gobble=4,basicstyle=\ttfamily\scriptsize]
    method acquire(p: Node, r: Node) {
      if (r.next == null) {
        r.next := p; update(p, -1, r.curr_prio)
      } else {
        p.next := r; update(r, -1, p.curr_prio)
      }
    }
    method update(n: Node, from: Int, to: Int) {
      n.prios := n.prios $\setminus$ {from}
      if (to >= 0) n.prios := n.prios $\cup$ {to}
      from := n.curr_prio
      n.curr_prio := max(n.prios $\cup$ {n.def_prio})@\label{line:pip-max-calc}@
      to := n.curr_prio;
      if (from != to && n.next != null) {
        update(n.next, from, to)
      }
    }
  \end{lstlisting}
  \end{minipage}%
  \begin{minipage}[c]{.44\linewidth}
    \begin{tikzpicture}[>=stealth, scale=0.8, every node/.style={scale=0.8}]
      \def\xsep{1.5}
      \def\ysep{1.35}

      % Nodes
      \node[unode] (p1) {$p_1$};
      \node[prio] (priop1) at ($(p1) + (-.6, 0)$) {$\mathbf{\underline{\color{blue} 3}}$};
      \node[prio] (priosp1) at ($(p1) + (.65, 0)$) {$\{2\}$};

      \node[rnode] (r1) at ($(p1) + (\xsep, \ysep)$) {$r_1$};
      \node[prio] (prior1) at ($(r1) + (0, .6)$) {$\mathbf{\underline{\color{blue} 0}}$};
      \node[prio] (priosr1) at ($(r1) + (0, -.6)$) {$\emptyset$};

      \node[unode] (p2) at ($(r1) + (1.5*\xsep, 0)$) {$p_2$};
      \node[prio] (priop2) at ($(p2) + (0, .6)$) {$\mathbf{\underline{\color{blue} 1}}$};
      \node[prio] (priosp2) at ($(p2) + (0, -.6)$) {$\{0\}$};

      \node[rnode] (r2) at ($(p2) + (\xsep, -\ysep)$) {$r_2$};
      \node[prio] (prior2) at ($(r2) + (-.6, 0)$) {$0$};
      \node[prio] (priosr2) at ($(r2) + (.65, 0)$) {$\{\mathbf{\underline{\color{blue} 1}}\}$};

      \node[unode] (p3) at ($(r2) + (-\xsep, -\ysep)$) {$p_3$};
      \node[prio] (priop3) at ($(p3) + (0, .6)$) {$\mathbf{\underline{\color{blue} 2}}$};
      \node[prio] (priosp3) at ($(p3) + (0, -.6)$) {$\{1\}$};

      \node[rnode] (r3) at ($(p1) + (\xsep, -\ysep)$) {$r_3$};
      \node[prio] (priosr3) at ($(r3) + (0, -.6)$) {$\{1,\mathbf{\underline{\color{blue} 2}},2\}$};
      \node[prio] (prior3) at ($(r3) + (0, .6)$) {$0$};

      \node[unode] (p4) at ($(p1) + (0, -2*\ysep)$) {$p_4$};
      \node[prio] (priop4) at ($(p4) + (-.6, 0)$) {$\mathbf{\underline{\color{blue} 1}}$};
      \node[prio] (priosp4) at ($(p4) + (.65,0)$) {$\emptyset$};

      \node[unode] (p5) at ($(p4) + (2*\xsep, 0)$) {$p_5$};
      \node[prio] (priop5) at ($(p5) + (-.6, 0)$) {$\mathbf{\underline{\color{blue} 2}}$};
      \node[prio] (priosp5) at ($(p5) + (.55,0)$) {$\emptyset$};

      \node[rnode] (r4) at ($(p5) + (0, -\ysep)$) {$r_4$};
      \node[prio] (priosr4) at ($(r4) + (.9, 0)$) {$\{\mathbf{\underline{\color{blue} 2}},2\}$};
      \node[prio] (prior4) at ($(r4) + (-.6, 0)$) {$0$};

      \node[unode] (p6) at ($(p5) + (1.5*\xsep, 0)$) {$p_6$};
      \node[prio] (priop6) at ($(p6) + (-.6, 0)$) {$\mathbf{\underline{\color{blue} 2}}$};
      \node[prio] (priosp6) at ($(p6) + (.55,0)$) {$\emptyset$};

      \node[unode] (p7) at ($(p4) + (0, -\ysep)$) {$p_7$};
      \node[prio] (priop7) at ($(p7) + (-.6, 0)$) {$\mathbf{\underline{\color{blue} 2}}$};
      \node[prio] (priosp7) at ($(p7) + (.55,0)$) {$\emptyset$};

      \draw[fedge] (p1) to[bend left=30] (r1);
      \draw[edge] (r1) to (p2);
      \draw[edge] (p2) to[bend left=30] (r2);
      \draw[edge] (r2) to[bend left=30] (p3);
      \draw[edge] (p3) to (r3);
      \draw[edge] (r3) to[bend left=30] (p1);
      \draw[edge] (p4) to[bend left=30] (r3);
      \draw[edge] (p5) to[bend right=30] (r3);
      \draw[edge] (p6) to (r4);
      \draw[edge] (p7) to[bend left=30] (r4);
    \end{tikzpicture}
  \end{minipage}
  \caption{Pseudocode of the PIP and a state of the protocol data structure. Round nodes
    represent processes and rectangular nodes resources. Nodes are
    marked with their \tw{default priorities \lstinline+def_prio+} as
    well as the aggregate priority multiset \lstinline+prios+. A
    node's \tw{current priority \lstinline+curr_prio+} is underlined and marked in bold blue.\label{fig-pip}}
  \vspace*{-1.5em}
\end{figure}

As a motivating example, we consider an idealized priority inheritance protocol
(PIP), which is a technique used in process
scheduling~\cite{DBLP:journals/tc/ShaRL90}. The purpose of the
protocol is to avoid (unbounded) priority inversion, i.e., a situation
where a low-priority process blocks a high-priority process from making progress.
The protocol maintains a bipartite graph with nodes representing
processes and resources. An example graph is shown in
Fig.~\ref{fig-pip}. An edge from a process $p$ to a resource $r$
indicates that $p$ is waiting for $r$ to become available whereas an
edge in the other direction means that $r$ is currently held by
$p$. Every node has an associated default priority as well as a
current priority, both of which are natural numbers. The
current priority affects scheduling decisions. When a process attempts
to acquire a resource currently held by another process, the graph is
updated to avoid priority inversion. For example, when process $p_1$
with current priority $3$ attempts to acquire the resource $r_1$ that
is held by process $p_2$ of priority $2$, then $p_1$'s higher priority
is propagated to $p_2$ and, transitively, to any other process that
$p_2$ is waiting for ($p_3$ in this case). As a result, all nodes on
the created cycle will be updated to current priority $3$\footnote{The
  algorithm can then detect the cycle to prevent a deadlock, but this is not the concern of this data structure.}. The
protocol thus maintains the following
\emph{invariant}: % invariant that
the current priority of each node is the maximum of its default
priority and the current priorities of all its predecessors. Priority
propagation is implemented by the method \lstinline+update+ shown in
Fig~\ref{fig-pip}. The implementation represents graph edges by
\lstinline+next+ pointers and handles both kinds of modifications to
the graph: adding an edge (\lstinline+acquire+) and removing an edge
(\lstinline+release+ - code omitted). To recalculate the current
priority of a node (line~\ref{line:pip-max-calc}), each node maintains
its default priority \lstinline+def_prio+ and a multiset
\lstinline+prios+ which contains the priorities of all its immediate
predecessors.% as well as its own default priority.

Verifying that the PIP maintains its invariant using established
separation logic (SL) techniques is challenging. In general, SL assertions
describe resources and express the fact that the program has
permission to access and manipulate these resources. \tw{In what follows, w}e stick to the
standard model of SL where resources are memory regions represented as
partial heaps. \tw{We sometimes view partial heaps more
abstractly as partial graphs (hereafter, simply graphs)}. Assertions describing larger regions are built from
smaller ones using \emph{separating conjunction}, $\phi_1 *
\phi_2$. Semantically, the $*$ operator is tied to a notion of
resource composition defined by an underlying \emph{separation
  algebra}~\cite{DBLP:conf/lics/CalcagnoOY07,
  DBLP:conf/aplas/CaoCA17}. In the standard model, composition enforces
that $\phi_1$ and $\phi_2$ must describe disjoint regions. The logic
and algebra are set up so that changes to the region $\phi_1$ do not
affect $\phi_2$ (and vice versa). That is, if $\phi_1 * \phi_2$ holds
before the modification and $\phi_1$ is changed to $\phi_1'$, then
$\phi_1' * \phi_2$ holds afterwards. This so-called \emph{frame rule}
enables modular reasoning about modifications to the heap and extends
well to the concurrent setting when threads operate on disjoint
portions of memory~\cite{DBLP:journals/siglog/BrookesO16, Dodds:2016:VCS:2866613.2818638,
  DBLP:conf/esop/RaadVG15, DBLP:conf/aplas/DockinsHA09}. However, the
mere fact that $\phi_2$ is preserved by modifications to $\phi_1$ does
not guarantee that if a global property such as the PIP invariant holds
for $\phi_1 * \phi_2$, it also still holds for $\phi_1' *
\phi_2$.

For example, consider the PIP scenario depicted in
Fig.~\ref{fig-pip}. If $\phi_1$ describes the subgraph containing only
node $p_1$, $\phi_2$ the remainder of the graph, and $\phi_1'$ the
graph obtained from $\phi_1$ by adding the edge from $p_1$ to $r_1$,
then the PIP invariant will no longer hold for the new composed graph
described by $\phi_1' * \phi_2$. On the other hand, if $\phi_1$
captures $p_1$ and the nodes reachable from $r_1$ (i.e., the set of nodes modified by %footprint of
\lstinline+update+), $\phi_2$ the remainder of the graph, and we
reestablish the PIP invariant locally in $\phi_1$ obtaining $\phi_1'$
(i.e., run \lstinline+update+ to completion), then $\phi_1' * \phi_2$
will also globally satisfy the PIP invariant. The separating conjunction
$*$ is not sufficient to differentiate these two cases; both describe valid
partitions of a possible program heap.
%\tw{In essence, the implicit
%  interface of a partial graph that determines whether
 As a consequence, prior techniques have
to revert back to non-local reasoning to prove that the invariant is
maintained.

\tw{
A first helpful idea towards a solution of this problem is that of \emph{iterated separating
  conjunction}~\cite{Yang01ShorrWaite, DBLP:conf/cav/0001SS16}, which
describes a graph $G$ consisting of a set of nodes $X$ by a
formula $\Psi = \Sep_{x \in X} \nodePred(x)$ where $\nodePred(x)$ is
some predicate that holds locally for every node $x \in X$. Using such
node-local conditions one can naturally express non-inductive
properties of graphs (e.g. \emph{``$G$ has no outgoing edges''} or
\emph{``$G$ is bipartite''}). The advantage of this style of specification is two-fold. First, one can arbitrarily decompose and
recompose $\Psi$ by splitting $X$ into disjoint subsets. For example,
if $X$ is partitioned into $X_1$ and $X_2$, then $\Psi$ is equivalent
to $\Sep_{x \in X_1} \nodePred(x) * \Sep_{x \in X_2}
\nodePred(x)$. Moreover, it is very easy to prove that $\Psi$ is preserved under
modifications of subgraphs. For instance, if a program modifies the subgraph induced by
$X_1$ such that $\Sep_{x \in X_1} \nodePred(x)$ is preserved locally, then
the frame rule guarantees that $\Psi$ will be preserved in
the new larger graph. Iterated separating conjunction thus yields a simple
proof technique for local reasoning about graph properties that can be
described in terms of node-local conditions. However, this idea alone does not actually solve our problem because
general global graph properties such as \emph{``$G$ is a direct acyclic graph''},
\emph{``$G$ is an overlay of multiple trees''}, or \emph{``$G$ satisfies the PIP invariant''} cannot be directly described this way.

% In this paper we adapt this proof technique to reason about general global
% graph properties such as \emph{``$G$ is a direct acyclic graph''},
% \emph{``$G$ is an overlay of multiple trees''}, or \emph{``$G$
%   satisfies the PIP invariant''}.
% %
% To this end we need to answer two interrelated questions: (1) how can
% we encode global graph properties in terms of node-local conditions;
% and (2) what is the right separation algebra underlying iterated
% separating conjunction and what is the accompanying notion
% of interface of a partial graph that enables one to prove locally that
% global graph properties are preserved by using separation-logic-style framing?

\paragraph{Solution.}
The key ingredient of our approach is the concept of a \emph{flow}
  of a graph: a function $\flow$ from the nodes of the graph to \emph{flow values}.
% In this paper we answer the question: What %is the
%  separation algebra
% %that
%  allows us to reason modularly about the effects of local changes
% on global properties of graphs? We consider a general class of global
% graph properties that \sk{can be expressed in terms of \emph{flows} --
%   functions from nodes of the graph to values}.
For the PIP, the flow maps
each node to the multiset of its incoming priorities. In
general, a flow is a fixpoint of a set of algebraic equations induced
by the graph. These equations are defined over a \emph{flow domain},
which determines how flow values are propagated along the edges
of the graph and how they are aggregated at each node. In the PIP example,
an edge between nodes $(n, n')$ propagates the multiset containing $\max(\flow(n),n.\mathtt{def\_prio})$ from $n$ to $n'$. The multisets arriving at $n'$ are aggregated with multiset union to obtain $\flow(n')$.
Flows enable capturing global graph properties in terms of node-local conditions. For example, the PIP invariant can be expressed by the following node-local condition:
$n.\mathtt{curr\_prio} = \max(\flow(n),n.\mathtt{def\_prio})
%  \enspace.
%$
$.
To enable compositional reasoning about such properties we need an appropriate separation algebra allowing us to prove locally that modifications to a
subgraph do not affect the flow of the remainder of the graph.

To this end, we make the useful observation that a separation algebra
induces a notion of an \emph{interface of a resource}: we say that two
resources $a$ and $a'$ are equivalent if they compose with the same
resources. The interface of a resource $a$ is then given by $a$'s
equivalence class.  In the standard model of SL where resources are
graphs and composition is disjoint graph union, the interface of a graph $G$ is the
set of all graphs $G'$ that have the same domain as $G$.
% That is, in this model we can
% identify $G$'s interface with $G$'s domain.

The interfaces of resources described by assertions capture the information that is implicitly communicated when these assertions are conjoined by
separating conjunction. As we discussed earlier, in the standard model of SL, this
information is too weak to enable local reasoning about global
properties of the composed graphs because some additional information
about the subgraphs' structure other than which nodes they contain
must be communicated. For instance, if the goal is to verify the PIP
invariant, the interfaces must capture information about the multisets of
priorities propagated between the subgraphs. We define a separation algebra achieving exactly this: the induced \emph{flow interface} of a graph $G$ in this separation algebra captures how values of the flow domain must enter and leave
$G$ such that, when composed with a compatible
graph $G'$, the imposed local conditions on the flow of each node are satisfied
in the composite graph.

This is the key to enabling SL-style framing for global graph properties.
Using iterated separating conjunctions over the new separation algebra, we obtain a compositional proof technique that yields succinct proofs of programs such as the PIP, whose proofs with existing techniques would involve non-trivial global reasoning steps.

}

% interface of a partial graph $G_1$ that
% abstracts from $G_1$'s internal structure while capturing enough
% information to prove, using separation-logic-style framing, that global
% properties of a larger graph containing $G$ are maintained by
% modifications of $G_1$?}

\paragraph{Contributions.}
In \rSc{sec-framework}, we present mathematical foundations for
flow domains, imposing minimal requirements on the underlying algebra
that allow us to capture a broad range of data structure
invariants and graph properties, and reason locally about them in a
suitable separation algebra. Building on this theory we develop a general proof technique for
modular reasoning about global graph properties that can be integrated
with existing separation logics \tw{provided they support iterated separating
conjunction based on the standard heap separation
algebra}~(\rSc{sec-proof-technique}).
% \tw{The key idea here is to split
% the specification into a conjunction of two iterated separating
%   conjunctions. The first one, expressed in terms of the actual separating
%   conjunction of the logic, abstracts a heap region by a collection of
%   singleton flow interfaces. The second one, expressed in the pure
%   part of the logic in terms of
%   our flow separation algebra, composes
%   these node-level flow interfaces to a single flow interface of the
%   abstracted heap region.}
We further identify general mathematical conditions
that guarantee unique flows and provide local proof arguments to check
the preservation of these conditions~(\rSc{sec-ea}).
 We
demonstrate the versatility of our approach by presenting local
proofs for two challenging examples: the
PIP and the concurrent
non-blocking list due to Harris~\cite{DBLP:conf/wdag/Harris01}.

\iffalse
We further devise a strategy for automating this technique using
SMT-based verification tools (\rSc{sec-automation}). We have
implemented this strategy on top of the verification tool
Viper~\cite{DBLP:series/natosec/0001SS17} and applied it successfully
to a variety of challenging benchmarks (\rSc{sec-evaluation}).  These
include 1) algorithms involving general graphs such as the PIP and
Dijkstra's algorithm, 2) inductive structures such as linked lists and
B trees, 3) overlaid structures such as the Harris list with
draining~\cite{DBLP:conf/wdag/Harris01} and threaded trees, and 4) OO
design patterns such as Composite and Subject/Observer. We are not
aware of any other approach that can handle these examples with the
same degree of simplicity or automation.
\fi

\paragraph{Flows Redesigned.}
Our work is inspired by the recent flow framework explored by some of
the authors~\cite{DBLP:journals/pacmpl/KrishnaSW18}. We revisit the
core algebra behind flows reasoning, and derive a different algebraic
foundation by analysing the minimal requirements for general local
reasoning; we call our newly-designed reasoning framework the \emph{\framework}.
Our new mathematical foundation makes several
significant improvements over~\cite{DBLP:journals/pacmpl/KrishnaSW18} and eliminates its
most stark limitations.
First, we present a simplified and generalized meta theory of flows
that makes the approach much more broadly applicable. \tw{For example, the original
framework cannot reason locally about certain graph updates such as
removing an edge that breaks a cycle (which can happen in the
PIP). Our new framework provides an elegant solution to this problem
by requiring that the aggregation operation on flow values is
cancellative (see \rSc{sec-flows}). This requirement is
fundamentally incompatible with the algebraic foundation of the
original framework, thus necessitating our new development. We show
that requiring cancellativity does not limit expressivity. Moreover, the new framework is much more convenient to
use because, unlike the original framework, it imposes no restrictions on how flow values are propagated along
edges in the graph.}
Next, the proofs of programs shown
in~\cite{DBLP:journals/pacmpl/KrishnaSW18} %are only on paper and
depend on a bespoke program logic. This logic requires new reasoning
primitives that are not supported by the logics implemented in
existing SL-based verification tools. Our general proof technique
eliminates the need for a dedicated program logic and can be implemented
on top of standard separation logics and existing SL-based
tools.
Finally, the underlying separation
algebra of the original framework makes it hard to use equational
reasoning, which is a critical prerequisite for enabling proof
automation.
%We demonstrate that the simplified separation algebra
%paves the way for automated reasoning about flows using SMT solvers.
%
We provide a more detailed technical comparison
to~\cite{DBLP:journals/pacmpl/KrishnaSW18} and other related work
in~\S\ref{sec-related}.

%%% Local Variables:
%%% mode: latex
%%% TeX-master: "writeup"
%%% End:

%\section{Flows from First Principles}\label{sec-principles}
%\input{principles}
%\section{The \FRAMEWORK}\label{sec-framework}
%\input{framework}
\section{The \FRAMEWORK}\label{sec-principles}\label{sec-framework}
In this section, we introduce the \framework, explaining the motivation for its design with respect to local reasoning principles. We aim for a general technique for modularly proving the preservation of recursively-defined invariants over (partial) graphs, with well-defined decomposition and composition operations.
% \as{Partial graphs (hereafter, simply graphs), enriched with flow
%   information according to our framework, form a separation algebra
%   (\secref{framework}), allowing our technique to integrate seamlessly
%   with a variety of separation-logic-based proof techniques, as we
%   will show in \secref{notsurewhich}.}\asfootnote{This may be
%   redundant with respect to the intro, as Thomas commented for a
%   similar sentence last time around. However, I found the prior
%   sentence here (about advantages of partialness) a bit
%   unclear/imprecise. We could also shorten.}
%We begin with some basic definitions and notations that we use in the rest of this paper.

\subsection{Preliminaries and Notation}
\label{sec-preliminaries}

The term $\ite{b}{t_1}{t_2}$ denotes $t_1$ if condition $b$ holds and $t_2$ otherwise.
We write $f\colon A \to B$ for a function from $A$ to $B$, and $f \colon A \pto B$ for a partial function from $A$ to $B$.
For a partial function $f$, we write $f(x) = \bot$ if $f$ is undefined at $x$.
We use lambda notation $\lambdaFn{x}{E}$ to denote a function that maps $x$ to the expression $E$ (typically containing $x$).
If $f$ is a function from $A$ to $B$, we write $f[x \goesto y]$ to denote the function from $A \cup \set{x}$ defined by $f[x \goesto y](z) \defeq \ite{z = x}{y}{f(z)}$.
We use $\set{x_1 \goesto y_1, \dotsc, x_n \goesto y_n}$ for pairwise different $x_i$ to denote the function $\emptyFn[x_1 \goesto y_1]\dotsm[x_n \goesto y_n]$, where $\emptyFn$ is the function on an empty domain.
Given functions $f_1 \colon A_1 \to B$ and $f_2 \colon A_2 \to B$ we write $f_1 \uplus f_2$ for the function $f \colon A_1 \uplus A_2 \to B$ that maps $x \in A_1$ to $f_1(x)$ and $x \in A_2$ to $f_2(x)$ (if $A_1$ and $A_2$ are not disjoint sets, $f_1 \uplus f_2$ is undefined).

We write $\deltaFn{n}{n'} \colon \mDom \to \mDom$ for the function defined by $\deltaFn{n}{n'}(\mVar) \defeq \mVar$ if $n = n'$ else $\mZero$.
We also write $\zeroFn \defeq \lambdaFn{\mVar}{\mZero}$ for the identically zero function, $\identityFn \defeq \lambdaFn{\mVar}{\mVar}$ for the identity function, and use $e \equiv e'$ to denote function equality.
For $e \colon \mDom \to \mDom$ and $\mVar \in \mDom$ we write $\mVar \pipe e$ to denote the function application $e(\mVar)$.
We write $e \fnComp e'$ to denote function composition, i.e. $(e \fnComp e')(\mVar) = e(e'(\mVar))$ for $\mVar \in \mDom$, and use superscript notation $e^p$ to denote the function composition of $e$ with itself $p$ times.

For multisets $S$, we use standard set notation when clear from the context. We write $S(x)$ to denote the number of occurrences of $x$ in $S$.
We write $\set{x_1 \goesto i_1, \dotsc, x_n \goesto i_n}$ for the multiset containing $i_1$ occurrences of $x_1$, $i_2$ occurrences of $x_2$, etc.

  A \emph{partial monoid} is a set $\mDom$, along with a partial binary operation $\mOp \colon \mDom \times \mDom \pto \mDom$, and a special zero element $\mZero \in \mDom$, such that
  \begin{enumerate*}[label=(\arabic{enumi})]
  \item $\mOp$ is associative, i.e., $(\mVar_1 \mOp \mVar_2) \mOp \mVar_3 = \mVar_1 \mOp (\mVar_2 \mOp \mVar_3)$; and
  \item $\mZero$ is an identity, i.e., $\mVar \mOp \mZero = \mZero \mOp \mVar = \mVar$.
  \end{enumerate*}
  Here, $=$ means either both sides are defined and equal, or both are undefined.
We identify a partial monoid with its support set $\mDom$.
If $\mOp$ is a total function, then we call $\mDom$ a monoid.
Let $\mVar_1, \mVar_2, \mVar_3 \in \mDom$ be arbitrary elements of the (partial) monoid in the following.
We call a (partial) monoid $\mDom$ \emph{commutative} if $\mOp$ is commutative, i.e., $\mVar_1 \mOp \mVar_2 = \mVar_2 \mOp \mVar_1$.
Similarly, a commutative monoid $\mDom$ is \emph{cancellative} if $\mOp$ is cancellative, i.e., if $\mVar_1 \mOp \mVar_2 = \mVar_1 \mOp \mVar_3$ is defined, then $\mVar_2 = \mVar_3$.

%\begin{definition}[Separation Algebra~\cite{DBLP:conf/lics/CalcagnoOY07}]
  A \emph{separation algebra}~\cite{DBLP:conf/lics/CalcagnoOY07} is a cancellative, partial commutative monoid.
%\end{definition}

\subsection{Flows}
\label{sec-flows}

\as{Recursive properties of graphs naturally depend on non-local information; \eg{} we cannot express that a graph is acyclic directly as a conjunction of per-node invariants.
Our \framework{} \as{captures non-local graph properties by defining \emph{flow values} at each node; its entire theory }
is parametric with the choice of a \emph{flow domain}, whose components will be explained and motivated in this section.
We use \emph{two} running examples of graph properties to illustrate our explanations in this section:
\begin{enumerate}
\item Firstly, we consider \emph{path-counting}, defining a flow domain whose flow values at each node represent the number of paths to this node from a distinguished node $n$.
Path-counting provides enough information to express locally per node that e.g.~(a) all nodes are reachable from $n$ (the path count is non-zero), or (b) that the graph forms a tree rooted at $n$ (all path counts are exactly $1$).
\item Secondly, we use the PIP (\rF{fig-pip}), defining flows with which we can  locally capture the appropriate current node priorities as the graph is modified.
\end{enumerate}}

\begin{definition}[Flow Domain]
  A \emph{flow domain} $(\mDom, \mOp, \mZero, \edgeDom)$ consists of
  a commutative cancellative (total) monoid $(\mDom, \mOp, \mZero)$ and a set of \emph{edge functions} $\edgeDom \subseteq \mDom \to \mDom$.
\end{definition}
%\as{For ease of reference in the remainder of this section, we provide some immediate examples; their details will be explained below.}
\begin{example}
  \label{ex-fd-path-count}
  The flow domain used for the path-counting flow is $(\Nat, +, 0, \set{\identityFn, \zeroFn})$, consisting of the monoid on natural numbers under addition and the set of edge functions containing only the identity function and the zero function.
\end{example}

\begin{example}
  \label{ex-fd-pip}
We use $(\Nat^{\Nat}, \cup, \emptyset, \set{\zeroFn} \cup \set{\lambdaFn{m}{\set{\max(m\cup\set{p})}} \mid p{\in}\Nat})$\asfootnote{Renamed flow domain elements to $m$ and natural numbers to $p$, here.} as flow domain for the PIP example. This consists of the monoid of multisets of natural numbers under multiset union and two kinds of edge functions: $\zeroFn$ and functions mapping a multiset $m$ to the singleton multiset containing the maximum value between $m$ and a fixed value $p$ (representing a node's default priority).
\end{example}
\as{As explained below, edge functions are used to determine which flow values are propagated from node to node around the graph.
For further definitions in this section we will assume a fixed} flow domain $(\mDom, \mOp, \mZero, \edgeDom)$ and a (potentially infinite) set of nodes $\nodeDom$. \as{For this section, we abstract heaps using directed partial graphs; integration of our graph reasoning with direct proofs over program heaps is explained in \secref{proof-technique}}.
\begin{definition}[Graph]
  A \emph{(partial) graph} $\graph = (N, \edgeFn)$ consists of a finite set of nodes $N \subseteq \nodeDom$ and a mapping from pairs of nodes to edge functions $\edgeFn \colon N \times \nodeDom \to \edgeDom$.
\end{definition}
\paragraph{Flow Values and Flows}
\as{Flow values (taken from $\mDom$; the first element of a flow domain) are used to capture sufficient information to express desired non-local properties of a graph. In Example~\ref{ex-fd-path-count}, flow values are non-negative integers; for the PIP (Example~\ref{ex-fd-pip}) we instead use \emph{multisets} of integers, representing relevant \emph{non-local} information: the priorities of nodes currently referencing a given node in the graph. Given such flow values, a node's correct priority can be defined locally per node in the graph. This definition requires only the \emph{maximum} value of these multisets, but as we will see shortly these multisets enable local \emph{recomputation} of a correct priority when the graph is changed.}

\as{For a graph $\graph = (N, \edgeFn)$} we express properties of $\graph$ in terms of node-local conditions that may depend on the nodes' \emph{flow}. A flow is a function $\flow: N \to \mDom$ assigning every node a flow value and must be some fixpoint of the following \emph{flow equation}:
\begin{align}
%  \label{eq-flow}
  \label{eqn-flow-equation}
  \forall n \in N.\; \flow(n) = \inflow(n) \mOp \kern-3pt \mBigOp_{n' \in N} \flow(n') \pipe \edgeFn(n', n)
    \tag{\textsf{FlowEqn}}
\end{align}
Intuitively, one can think of the flow as being obtained by a fold computation over the graph\footnote{We note that flows are not generally defined in this manner as we consider any fixpoint of the flow equation to be a flow. \as{Nonetheless}, the analogy helps to build \as{an initial} intuition.}: the \emph{inflow} $\inflow: N \to \mDom$ defines an initial flow at each node. This initial flow is then updated recursively for each node $n$: the current flow value at its predecessor nodes $n'$ is transferred to $n$ via \emph{edge functions} $e(n',n): \mDom \to \mDom$
% TW: the following notation is already defined in the preliminaries
% (we use $\pipe$ to denote function application where the function is on the \emph{right}).
These flow values are aggregated using the \emph{summation operation} $+$ of the flow domain to obtain an updated flow of $n$\as{; a flow for the graph is some fixpoint satisfying this equation at all nodes.}

\begin{example}
  \twfootnote{I think this example needs to be moved further down since we haven't yet talked about the edge functions for the PIP graph.}
  Consider the graph in \figref{pip}; if the flow domain is as
  in \rE{ex-fd-pip}, the inflow function $\inflow$ assigns \as{the empty multiset to every node $n$}
  and we let $\flow(n)$ be
  the multiset labelling every node in the figure, then $\flowEqn(\inflow, \edgeFn,
  \flow)$ holds.
\end{example}

\begin{definition}[Flow Graph]
  A \emph{flow graph} $\fGraph = (N, \edgeFn, \flow)$ is a graph $(N, \edgeFn)$ and function $\flow \colon N \to \mDom$ such that there exists an \emph{inflow} $\inflow \colon N \to \mDom$ satisfying $\flowEqn(\inflow, \edgeFn, \flow)$.
\end{definition}
We let $\dom(\fGraph) = N$, and sometimes identify $\fGraph$ and $\dom(\fGraph)$ to ease notational burden.
For $n \in \fGraph$ we write $\fGraph_n$ for the singleton flow subgraph of $\fGraph$ induced by $n$.

\paragraph{Edge Functions}
\tw{In any flow graph, the flow value assigned to a node $n$ by a flow is propagated to its neighbours $n'$ (and transitively) according to the edge function $\edgeFn(n,n')$ labelling the edge $(n,n')$. The edge function maps} the flow value at the \emph{source node} $n$ to one propagated on \as{\emph{this edge}} to the \emph{target node} $n'$. Note that we require such a labelling for \emph{all} pairs consisting of a source node $n$ inside the graph and a target node \tw{$n' \in \nodeDom$ (i.e.,} possibly outside the graph). %\asout{In addition, we provide a convenient default case.}
\as{The $0$ flow value (the second element of our flow domains)} is used to represent no flow; the corresponding (constant) zero \emph{function} $\zeroFn = \lambdaFn{\mVar}{\mZero}$ \as{is used as edge function to model the \emph{absence}} of an edge in the graph\footnote{We will sometimes informally refer to \emph{paths} in a graph as meaning sequences of nodes for which no edge function labelling a consecutive pair in the sequence is the zero function $\zeroFn$.}.
% We write $\edgeFn(n,n')$ for the edge function labelling the pair $(n,n')$.
A set of edge functions $\edgeDom$ from which this labelling is chosen \as{can, other than the requirement $\zeroFn\in\edgeDom$, be chosen as desired. As} we will see in \secref{unique-fixpoints}, restrictions to \as{particular sets of edge functions} $\edgeDom$ can be exploited to \as{further} strengthen our overall technique.

For our PIP example, we choose the edge functions to be $\zeroFn$ where no edge exists in the PIP structure, and otherwise $\lambdaFn{X}{\set{\max(X \cup \set{m})}}$ where $m$ is the default priority of the source of the edge. For example, in \rF{fig-pip}, $\edgeFn(r_3, p_2) = \zeroFn$ and \tw{$\edgeFn(r_3, p_1) = \lambdaFn{X}{\set{\max(X \cup \set{0})}}$}.
Since the flow value at $r_3$ is $\set{1, 2, 2}$, the edge $(r_3, p_1)$ propagates the value $\set{2}$ to $p_1$\as{, correctly representing the current priority of $r_3$.} \as{Edge functions can depend on the local state of the source node (\eg{} default priorities here); dependencies from elsewhere in the graph must be represented by the node's flow.}
%In the path-counting example, the edge functions would be identity
%functions (edge present) and zero functions (edge absent). %For last-edge path counting, the non-zero edge functions will map empty multisets to singletons containing the source node, and otherwise return a multiset with the same input cardinality, containing only the source node.
%\sk{For inverse reachability, the (non-zero) edge function on the edge $(n, n')$ maps a multiset of sets $T$ to a new multiset $T'$ containing $S %\cup \set{n}$ for every $S \in T$ such that $n \not\in S$.
%Note, $T'$ does not contain sets $S \in T$ that contain $n$, because such sets correspond to cyclic paths in the graph, while this particular domain tracks only simple paths.}

\paragraph{Flow Aggregation and Inflows}
The flow value at a node is defined by those propagated to it \as{from each node in a graph} via edge functions, along with an additional \emph{inflow} value, explained here. Since multiple \as{non-zero flow values can be propagated to a single node, we require an} aggregation of these values, \as{for} which a binary $\mplus$ operator on flow values must be defined\as{:} the third element of our flow domains. \as{We require $\mplus$ to be commutative and associative, making this aggregation order-independent}. The $0$ flow value \techreport{(representing no flow) }must act as a unit \as{for} $\mplus$. For example, in the path-counting flow domain $\mplus$ means addition on natural numbers, while for the multisets employed for the PIP it means multiset union.

Each node in a flow graph has an \emph{inflow}, modelling contributions to its flow value which do \emph{not} come from inside the graph. Inflows plays two important roles: first, since our graphs are partial, they model contributions from nodes \emph{outside of the graph}. Second, inflow can be artificially added as a means of \as{specialising the computation of flow values to characterise specific graph properties. For example, in the path-counting domain, we give an inflow of $1$ to the node from which we are counting paths, and $0$ to all others.}
%
%For example, in our PIP domain, the inflow of each node will be the singleton multiset containing the node's default priority\asfootnote{Is this true any more?}. In our path counting domain, we can select the distinguished root node $n$ by giving it an inflow of $1$; we could also do this for multiple nodes, to count paths from each.% In the inverse reachability domain, we can employ multisets containing a single empty set, forcing paths from these nodes to be tracked by the flow computation.

The flow equation \rEq{eqn-flow-equation} defines the flow of a node $n$ to be the aggregation of flow values coming from other nodes $n'$ inside the graph (as given by the respective edge function $\edgeFn(n', n)$) as well as the inflow $\inflow(n)$. Preserving solutions to this equation \as{across updates to the graph structure} is a fundamental goal of our technique. \as{The following lemma (which relies on the fact that $+$ is required to be cancellative) states that any correct flow values uniquely determine appropriate inflow values:}
\begin{lemma}
  \label{lem-inflow-unique}
  Given a flow graph $(N, \edgeFn, \flow) \in \fGraphDom$, there exists a unique inflow $\inflow \colon N \to \mDom$ such that $\flowEqn(\inflow, \edgeFn, \flow)$.
\end{lemma}
\techreport{
\begin{proof}
  Suppose $\inflow$ and $\inflow'$ are two solutions to $\flowEqn(\_, \edgeFn, \flow)$.
  Then, for any $n$,
  \[
    \flow(n) = \inflow(n) \mOp \mBigOp_{n' \in \dom(\inflow)} \flow(n') \pipe \edgeFn(n', n)
    = \inflow'(n) \mOp \mBigOp_{n' \in \dom(\inflow')} \flow(n') \pipe \edgeFn(n', n)
  \]
  which, by cancellativity of the flow domain, implies that $\inflow(n) = \inflow'(n)$.
\end{proof}
}
We now turn to how \as{solutions of the flow equation} can be \as{preserved or appropriately updated under \emph{changes} to the underlying graph}.

\skfootnote{Should we start the next subsection ``Flow Graph Composition and Abstraction'' at this point?}

\paragraph{Graph Updates and Cancellativity}
\as{Given a flow graph with known flow and inflow values, suppose we \emph{remove} an edge from $n_1$ to $n_2$ (replacing the edge function with $\zeroFn$). For the same inflow, such an update will potentially affect the flow at $n_2$ and nodes to which $n_2$ (transitively) propagates flow.
Starting from the simple case that $n_2$ has no outgoing edges, we need to recompute a suitable flow at $n_2$. Knowing the old flow value (say, $m$) and the contribution $m' = \flow(n_1) \pipe \edgeFn(n_1, n_2)$ \emph{previously} provided along the removed edge, we know that the correct new flow value is some $m''$ such that $m' + m'' = m$. This constraint has a unique solution (and thus, we can unambiguously recompute a new flow value) exactly when the aggregation $+$ is \emph{cancellative}; we therefore \twout{make} \tw{made} cancellativity a \emph{requirement} on the $+$ of any flow domain.}

\as{Cancellativity intuitively enforces that the flow domain carries enough information to enable adaptation to local updates (in particular, removal of edges\footnote{As we will show in \secref{flow-interfaces}, an analogous problem for composition of flow graphs is also directly solved by this choice to force aggregation to be cancellative.}). Returning to the PIP example, cancellativity requires us to carry multisets as flow values rather than only the maximum priority value: $+$ cannot be a maximum operation, as this would not be cancellative. The resulting multisets (similarly to the \lstinline+prio+ fields in the actual code) provide the information necessary to recompute corrected priority values locally.
}
%
%knowing locally how to update the flow at $n_2$ requires that we can uniquely
%
%Consider that we take a flow graph along with a correct inflow, and obtain a modified graph different only in that a single pair of nodes $(n_1,n_2)$ has a different edge function. We are concerned with the question of whether and how we can change the flow values in the new graph (keeping the inflow unchanged) to satisfy the flow equation.
%
%Consider first the simple case that the target $n_2$ of the
%modified edge propagates no flow via edge functions ($\edgeFn(n_2,n') =
%\zeroFn$ for all $n'$); it may however receive additional flow from edge functions $\edgeFn(n'', n_2)$ coming from nodes $n''$ other than $n_1$.
For example, in the PIP graph shown in \rF{fig-pip}, removing the edge from $p_6$ to $r_4$% (i.e. setting it to the zero function $\zeroFn$)
\as{would} not affect the current priority of $r_4$ whereas if $p_7$ had current priority $1$ instead of $2$, then the current priority of $r_4$ would have to decrease. % For example, in the path-counting\asfootnote{TODO: check that we write this consistently this way} domain, changing $\edgeFn(n_1,n_2)$ from the identity function $\identityFn = \lambda m. m$ to the zero function $\zeroFn = \lambda m. 0$ may require the flow at $n_2$ to decrease. This adjustment must reflect the \emph{number of paths} which depended on this edge;
%For this reason, our PIP Domain aggregates multisets of incoming flow values, rather than having $\mPlus$ simply collapse these to their maximum. The multisets contain enough information to locally adjust the flow value when an edge is removed from the graph, whereas if we knew only the maximum and removed an edge $(p, r)$ which provided exactly this value, we could not decide whether or not to decrease the flow value of $r$ without some knowledge of all of $r$'s incoming edges.
\as{In either case, recomputing} the flow value for $r_4$ is simply a matter of subtraction (removing $\{2\}$ from the multiset at $r_4$); \as{cancellativity guarantees that our flow domains will always provide the information needed for this recomputation.}
%
%this exploits the property that this flow domain is \emph{cancellative} with respect to $\mPlus$, giving us a unique solution.
%\footnote{Note that cancellativity does not imply the existence of additive inverse elements in our flow domains (flow domains are monoids, but not necessarily groups). In fact, we typically employ flow domains in which no non-zero elements have such inverses.\tw{THOMAS: This footnote can also be omitted as the example of multisets makes this fact obvious.}}.
\as{Without} this property, the recomputation of a flow value for the target node $n_2$ \as{would, in general, entail recomputing the incoming flow values from all remaining edges from scratch. Cancellativity is also crucial for Lemma \ref{lem-inflow-unique} above, forcing uniqueness of inflows, given known flow values in a flow graph. This allows us to define natural but powerful notions of flow graph decomposition and recomposition.
}

\subsection{Flow Graph Composition and Abstraction}
\label{sec-flow-interfaces}

\as{Building towards the core of our reasoning technique, we now turn to the question of decomposition and recomposition of flow graphs.}
Two flow graphs with disjoint domains always compose to a graph, but this will only be a \emph{flow graph} if their flows are chosen consistently to admit a solution to the resulting flow equation
(i.e.~the flow graph composition operator defined below is \emph{partial}).

\begin{definition}[Flow Graph Algebra]
  \label{def-flow-graphs}
  The \emph{flow graph algebra} $(\fGraphDom, \fGraphComp, \fGraphEmpty)$ for the flow domain $(\mDom, \mOp, \mZero, \edgeDom)$ is defined by
  \begin{align*}
    \fGraphDom
    &\defeq \setcomp{
      (N, \edgeFn, \flow)}
      {(N, \edgeFn, \flow) \text{ is a flow graph}} \\
    (N_1, \edgeFn_1, \flow_1) \fGraphComp (N_2, \edgeFn_2, \flow_2)
    &\defeq
      \begin{cases}
        \fGraph
        & \fGraph = (N_1 \uplus N_2, \edgeFn_1 \uplus \edgeFn_2, \flow_1 \uplus \flow_2) \in \fGraphDom \\
        \bot & \text{otherwise}
      \end{cases} \\
    \fGraphEmpty &\defeq (\emptyset, \edgeFnEmpty, \flowEmpty)
  \end{align*}
  where $\edgeFnEmpty$ and $\flowEmpty$ are the edge functions and flow on the empty set of nodes $N = \emptyset$.
  We use $\fGraph$ to range over $\fGraphDom$.
\end{definition}
\as{Intuitively, two flow graphs compose to a flow graph if their contributions to each others' flow (along edges from one to the other) are reflected in the corresponding inflow of the other graph. For example, consider the subgraph from \rF{fig-pip} consisting of the single node $p_7$ (with $0$ inflow). This will compose with the remainder of the graph depicted only if this remainder subgraph has an inflow which, at node $r_4$ includes at least the multiset $\{2\}$, reflecting the propagated value from $p_7$.}

We use this \as{intuition} to extract an \emph{abstraction} of flow graphs which we call \emph{flow interfaces}.
Given a flow (sub)graph, its flow interface consists of the node-wise inflow and \emph{outflow} (being the flow contributions its nodes make to all nodes outside of the graph).
It is thus an abstraction that hides the flow values and edges \as{wholly \emph{inside} the flow graph}.
Flow graphs that have the same flow interface ``look the same'' to the external graph, as the same values are propagated inwards and outwards.

Our abstraction of flow graphs consists of two complementary notions.
\as{Recall that} \rL{lem-inflow-unique} implies that any flow graph has a unique inflow.
Thus we can define an inflow function that maps each flow graph $\fGraph = (N, \edgeFn, \flow)$ to the unique inflow $\inflowFn(\fGraph) \colon \fGraph \to \mDom$ such that $\flowEqn(\inflowFn(\fGraph), \edgeFn, \flow)$.
\as{Dually, we} define the \emph{outflow} of $\fGraph$ as the function $\outflowFn(\fGraph) \colon \nodeDom \setminus N \to \mDom$ defined by
\[ \outflowFn(\fGraph)(n) \defeq \mBigOp_{n' \in N} \flow(n') \pipe \edgeFn(n', n).\]

\begin{definition}[Flow Interface]
  A \emph{flow interface} is a tuple $\interface = (\inflow, \outflow)$ where $\inflow \colon N \to \mDom$ and $\outflow \colon \nodeDom \setminus N \to \mDom$ for some $N \subseteq \nodeDom$.
\end{definition}
Given a flow graph $\fGraph \in \fGraphDom$, its flow interface $\interfaceFn(\fGraph)$ is the tuple $\paren{\inflowFn(\fGraph), \outflowFn(\fGraph)}$ consisting of its inflow and its outflow.
\as{Returning to the previous example, if $\fGraph$ is the singleton subgraph consisting of node $p_7$ from \rF{fig-pip} with flow and edges as depicted, then $\interfaceFn(\fGraph) = (\lambda n.\emptyset, \lambda n.\ite{n{=}r_4}{\{2\}}{\emptyset})$.}

We write $\inflowOfInt{\interface}, \outflowOfInt{\interface}$ for the two components of the interface $\interface = (\inflow, \outflow)$.
We again identify $\interface$ and $\dom(\inflowOfInt{\interface})$ to ease notational burden.

This \as{abstraction}, while simple, turns out to be powerful enough to build a separation algebra over our \as{flow graphs, allowing them} to be decomposed, locally modified and recomposed in ways yielding all the local reasoning benefits of separation logics. In particular, for graph operations within a subgraph with a certain interface, we need to prove: (a) that the modified subgraph is still a flow graph (by checking that the flow equation still has a solution locally in the subgraph) and (b) that it satisfies the same interface (in other words, the effect of the modification on the flow is contained within the subgraph)\as{; the meta-level results for our technique then} justify that we can recompose the modified subgraph with any graph that the original could be composed with.

We define the \as{corresponding} \emph{flow interface algebra} as follows\asfootnote{Maybe we should add some text describing the details of this definition}:
\begin{definition}[Flow Interface Algebra]
  \label{def-flow-interfaces}
  \begin{align*}
    \interfaces
    &\defeq \setcomp{\interface}{\interface  \text{ is a flow interface}} \\
    \interfaceEmpty
    &\defeq \interfaceFn(\fGraphEmpty) \\
    \interface_1 \intComp \interface_2
    &\defeq
    \begin{cases}
      \interface
      & \interface_1 \cap \interface_2 = \emptyset \\
      & \quad {}
      \land \forall i \neq j \in \set{1, 2}, n \in \interface_i.\;
      \inflowOfInt{\interface_i}(n) = \inflowOfInt{\interface}(n) \mOp \outflowOfInt{\interface_j}(n) \\
      & \quad {}
      \land \forall n \not\in \interface.\;
      \outflowOfInt{\interface}(n) = \outflowOfInt{\interface_1}(n) \mOp \outflowOfInt{\interface_2}(n) \\
      \bot & \text{otherwise.}
    \end{cases}
  \end{align*}
\end{definition}

Flow interface composition is well defined because of cancellativity of the underlying flow domain (\as{it is also, exactly as} flow graph composition, partial).
\techreport{
The interfaces of a singleton flow graph containing $n$ capture the flow and the outflow values propagated by $n$'s edges:

\begin{lemma}
  \label{lem-singleton-interfaces}
  For any flow graph $\fGraph = (N, \edgeFn, \flow)$ and $n,n' \in N$,
  if $\edgeFn(n, n) = \zeroFn$ then
  $\inflowOfInt{\interfaceFn(\fGraph_n)}(n) = \flow(n)$ and
  $\outflowOfInt{\interfaceFn(\fGraph_n)}(n') = \flow(n) \pipe \edgeFn(n, n')$.
\end{lemma}
\begin{proof}
  Follows directly from \rEq{eqn-flow-equation} and the definition of outflow.
\end{proof}
}
We next show the key result for this abstraction: the ability for two flow graphs to compose depends only on their interfaces\as{;
flow} interfaces implicitly define a congruence relation on flow graphs.

\begin{lemma}
  \label{lem-interface-composition-congruence}
  $\interfaceFn(\fGraph_1) = \interface_1 \land \interfaceFn(\fGraph_2) = \interface_2 \impl \interfaceFn(\fGraph_1 \fGraphComp \fGraph_2) = \interface_1 \intComp \interface_2$.
\end{lemma}
\techreport{
\begin{proof}
  If $\fGraph_1 \fGraphComp \fGraph_2$ is defined and has interface $\interface$, then we show that $\interface_1 \intComp \interface_2$ is defined and equal to $\interface$.
  Let $\fGraph_i = (N_i, \edgeFn_i, \flow_i), \interface = (\inflow, \outflow), \interface_1 = (\inflow_1, \outflow_1)$, and $\interface_2 = (\inflow_2, \outflow_2)$.
  Since $\fGraph = \fGraph_1 \fGraphComp \fGraph_2 \in \fGraphDom$ and $\inflowFn(\fGraph) = \inflowOfInt{\interface} = \inflow$, we know by definition that $\forall i \neq j \in \set{1, 2}, n \in \fGraph_i,$
  \begin{align*}
    & \quad \quad
      \flow(n) = \inflow(n) \mOp \mBigOp_{n' \in \fGraph} \flow(n') \pipe \edgeFn(n', n) \\
    & \iff
      \flow_i(n) = \inflow(n) \mOp \mBigOp_{n' \in \fGraph} \flow(n') \pipe \edgeFn(n', n) \\
    & \iff
      \inflow_i(n) \mOp \mBigOp_{n' \in \fGraph_i} \edgeFn_i(n', n, \flow_i(n')) \\
    & \quad = \inflow(n) \mOp \mBigOp_{n' \in \fGraph_i} \edgeFn_i(n', n, \flow_i(n'))  \mOp \mBigOp_{n' \in \fGraph_j} \edgeFn_j(n', n, \flow_j(n')) \\
    & \iff
      \inflow_i(n) = \inflow(n) \mOp \mBigOp_{n' \in \fGraph_j} \edgeFn_j(n', n, \flow_j(n'))  \tag{By cancellativity} \\
    & \iff
      \inflow_i(n) = \inflow(n) \mOp \outflow_j(n).
  \end{align*}
  Secondly, let $\fGraph = \fGraph_1 \fGraphComp \fGraph_2$ and note that
  \begin{align*}
    \outflow(n)
    &\defeq \mBigOp_{n' \in \fGraph} \flow(n') \pipe \edgeFn(n', n) \\
    &=  \mBigOp_{n' \in \fGraph_1} \flow_1(n') \pipe \edgeFn_1(n', n)
    \mOp  \mBigOp_{n' \in \fGraph_2} \flow_2(n') \pipe \edgeFn_2(n', n) \\
    &= \outflow_1(n) \mOp \outflow_2(n).
  \end{align*}
  As $\fGraph = \fGraph_1 \fGraphComp \fGraph_2$ implies $\dom(\fGraph_1) \cap \dom(\fGraph_2) = \emptyset$, this proves that $\interface_1 \intComp \interface_2 = \interface$.

  Conversely, if $\interface_1 \intComp \interface_2$ is defined and equal to $\interface$ then we show that  $\fGraph_1 \fGraphComp \fGraph_2$ is defined and has interface $\interface$.
  First, $\interface_1 \cap \interface_2 = \emptyset$, so we know that the graphs are disjoint.
  Note that the proof above works in both directions, so
  \begin{align*}
    & \quad \forall i \neq j \in \set{1, 2}, n \in \interface_i.\;
    \inflowOfInt{\interface_i}(n) = \inflowOfInt{\interface}(n) \mOp \outflowOfInt{\interface_j}(n) \\
    & \impl \flow(n) = \inflow(n) \mOp \mBigOp_{n' \in \fGraph} \flow(n') \pipe \edgeFn(n', n).
  \end{align*}
  This tells us that $\fGraph = \fGraph_1 \fGraphComp \fGraph_2 \in \fGraphDom$ and $\inflowFn(\fGraph) = \inflow$.
  From above, we also know that $\outflow(n) = \outflow_1(n) \mOp \outflow_2(n)$, so the interface composition condition $\outflowOfInt{\interface}(n) = \outflowOfInt{\interface_1}(n) \mOp \outflowOfInt{\interface_2}(n)$ gives us $\outflowFn(\fGraph) = \outflow$.
\end{proof}
}

\paragraph{Flow Footprints}
%The cancellativity of $\mPlus$ alone is not sufficient to reason locally about general flow graph modifications.
\sk{Consider again the simple modification of changing the edge function labelling a single edge $(n_1, n_2)$\as{; recall that we previously considered the simplified case above that $n_2$ has no outgoing edges.}}
\as{Cancellativity of $\mPlus$} avoids \as{a recomputation over} arbitrary \emph{incoming} edges, but
once we remove \as{this assumption, we} also need to account for the propagation of the change transitively throughout the graph.
For example, \as{by adding} the edge $(p_1,r_1)$ in \rF{fig-pip} and hence, $3$ to the flow of $r_1$, we \as{in turn} add $3$ to the flow of all other nodes reachable from $r_1$. On the other hand, adding an edge from $r_4$ to $p_5$ affects \tw{only the flow value of $p_5$}.
% For example, in the path counting domain, the addition or removal of a path affects the path count of all nodes reachable via (non-zero) edge functions. On the other hand, other operations such as swapping the order of two nodes on a simple path can be performed without affecting the flow values of other graph nodes.
To capture the relative locality of the side-effects of such updates,\as{ we introduce \emph{flow footprints}}. % of a modification to a flow graph.
A modification's flow footprint is the \emph{smallest subset} of the graph containing those nodes which are \as{\emph{sources}} of modified edges, plus all those whose flow values need to be changed in order to obtain a new flow graph \as{with an unchanged inflow}.
For example, the flow footprint for the addition of the edge $(p_1,r_1)$ in \rF{fig-pip} is $p_1$ and all nodes reachable from $r_1$ (including $r_1$ itself). \as{On the other hand, the flow footprint for removing the edge $(p_2,r_2)$ is just these two nodes; the flow to and from the rest of the graph remains unchanged. %Flow footprints can be used to localise the effect of a graph modification: from the perspective of the graph \emph{outside} of the flow footprint, nothing observable in the flow domain has changed.
We will exploit this idea to define
% abstractions and composition of flow graphs, ultimately enabling us to define
when a subgraph can be replaced with another without disturbing its surroundings.}

We next make this notion of flow footprint formally precise.

\begin{definition}[Flow Footprint]
  Let $\fGraph$ and $\fGraph'$ be flow graphs such that
  $\interfaceFn(\fGraph)=\interfaceFn(\fGraph')$, then the \emph{flow
    footprint} of $\fGraph$ and $\fGraph'$, denoted
  $\footprintFn(\fGraph, \fGraph')$, is the smallest flow graph
  $\fGraph_1'$ such that there exists $\fGraph_1,\fGraph_2$ with
  $\fGraph = \fGraph_1 \fGraphComp \fGraph_2$,
  $\fGraph' = \fGraph_1' \fGraphComp \fGraph_2$ and
  $\interfaceFn(\fGraph_1)=\interfaceFn(\fGraph_1')$.
\end{definition}

The following lemma states that the flow footprint captures exactly
those nodes in the graph that are affected by a modification
(i.e. either their flow or their outgoing edges change).

\begin{lemma}
  \label{lem-flow-footprint}
  Let $\fGraph$ and $\fGraph'$ be flow graphs such that
  $\interfaceFn(\fGraph)=\interfaceFn(\fGraph')$, then for all
  $n \in \fGraph$, $n \in \footprintFn(\fGraph,\fGraph')$ iff
  $\fGraph_n \neq \fGraph'_n$.
\end{lemma}

Crucially, the following result shows that we can use flow interfaces as an abstraction compatible with separation-logic-style framing.

\begin{theorem}
  \label{thm-flow-interfaces-sa}
  The flow interface algebra $(\interfaces, \intComp, \interfaceEmpty)$ is a separation algebra.
\end{theorem}
\techreport{
\begin{proof}
  We prove commutativity first, as it is used in the proof of associativity:
  \begin{itemize}
  \item $\intComp$ is commutative:\\
    This follows from the symmetry in the definition of $\intComp$ and the commutativity of the flow domain operator $\mOp$.
  \item $\intComp$ is associative, i.e. $\interface_1 \intComp (\interface_2 \intComp \interface_3) = (\interface_1 \intComp \interface_2) \intComp \interface_3$:\\
    Note that if any two of the interfaces $\interface_1, \interface_2$, and $\interface_3$ are not disjoint, then both sides of the equation are equal to $\bot$.
    We now show that if the LHS is defined, then the RHS is also defined and equal to it.
    Let $\interface_{23} = \interface_2 \intComp \interface_3$, and $\interface = \interface_1 \intComp \interface_{23}$.
    Define $\interface_{12} = \paren{\lambdaFn{n}{\inflowOfInt{\interface}(n) \mOp \outflowOfInt{\interface_3}(n)}, \lambdaFn{n'}{\outflowOfInt{\interface_1}(n') \mOp \outflowOfInt{\interface_2}(n)}}$.
    We first show that $\interface_{12} = \interface_1 \intComp \interface_2$.
    We know that the interfaces are disjoint, so let us check that the inflows are compatible.
    For $n \in \interface_1$,
    \begin{align*}
      \inflowOfInt{\interface_{12}}(n) \mOp \outflowOfInt{\interface_2}(n)
      &= \inflowOfInt{\interface}(n) \mOp \outflowOfInt{\interface_3}(n)
        \mOp \outflowOfInt{\interface_2}(n) \\
      &= \inflowOfInt{\interface}(n) \mOp \outflowOfInt{\interface_{23}}(n) \tag{As $\interface_{23} = \interface_2 \intComp \interface_3$} \\
      &= \inflowOfInt{\interface_1}(n). \tag{As $\interface = \interface_1 \intComp \interface_{23}$}
    \end{align*}
    On the other hand, for $n \in \interface_2$,
    \begin{align*}
      \inflowOfInt{\interface_{12}}(n) \mOp \outflowOfInt{\interface_1}(n)
      &= \inflowOfInt{\interface}(n) \mOp \outflowOfInt{\interface_3}(n)
        \mOp \outflowOfInt{\interface_1}(n) \\
      &= \inflowOfInt{\interface_{23}}(n) \mOp \outflowOfInt{\interface_3}(n) \tag{As $\interface = \interface_{23} \intComp \interface_1$} \\
      &= \inflowOfInt{\interface_2}(n). \tag{As $\interface_{23} = \interface_2 \intComp \interface_3$}
    \end{align*}
    Finally, the condition on the outflow follows by definition of $\interface_{12}$.

    We now show that $\interface = \interface_{12} \intComp \interface_3$.
    Again, the disjointness condition is satisfied by assumption, so let us check the inflows.
    If $n \in \interface_{12}$, $\inflowOfInt{\interface}(n) \mOp \outflowOfInt{\interface_3}(n) = \inflowOfInt{\interface_{12}}(n)$ by definition of $\interface_{12}$.
    And if $n \in \interface_3$,
    \begin{align*}
      \inflowOfInt{\interface}(n) \mOp \outflowOfInt{\interface_{12}}(n)
      &= \inflowOfInt{\interface}(n) \mOp \outflowOfInt{\interface_1}(n)
        \mOp \outflowOfInt{\interface_2}(n) \\
      &= \inflowOfInt{\interface_{23}}(n) \mOp \outflowOfInt{\interface_2}(n) \tag{As $\interface = \interface_{23} \intComp \interface_1$} \\
      &= \inflowOfInt{\interface_3}(n). \tag{As $\interface_{23} = \interface_2 \intComp \interface_3$}
    \end{align*}
    The condition on outflows is true because
    \begin{align*}
      \outflowOfInt{\interface}(n)
      &= \outflowOfInt{\interface_1}(n) \mOp \outflowOfInt{\interface_{23}}(n) \\
      &= \outflowOfInt{\interface_1}(n) \mOp \outflowOfInt{\interface_2}(n)
        \mOp \outflowOfInt{\interface_3}(n) \\
      &= \outflowOfInt{\interface_{12}}(n) \mOp \outflowOfInt{\interface_3}(n).
    \end{align*}

    Thus, if the LHS is defined and equal to $\interface$, then the RHS is defined and equal to $\interface$.
    By symmetry, and commutativity, the other direction is true as well.
  \item $\interfaceEmpty$ is an identity with respect to $\intComp$:\\
    This follows directly from the definitions.
  \item $\intComp$ is cancellative, i.e. $\interface_1 \intComp \interface_2 = \interface_1 \intComp \interface_3 \impl \interface_2 = \interface_3$:\\
    Let $\interface = \interface_1 \intComp \interface_2 = \interface_1 \intComp \interface_3$.
    Since the domains of $\interface_2$ and $\interface_3$ must be disjoint from $\interface_1$ and yet sum to the domain of $\interface$, they must be equal.
    Now for $n \in \interface_2$, by the definition of $\intComp$,
    \[
      \inflowOfInt{\interface_2}(n) = \inflowOfInt{\interface}(n) \mOp \outflowOfInt{\interface_1}(n)
      = \inflowOfInt{\interface_3}(n),
    \]
    so the inflows are equal.
    As for the outflows, if $n \not\in \interface$ then
    \[
      \outflowOfInt{\interface}(n) = \outflowOfInt{\interface_1}(n) \mOp \outflowOfInt{\interface_2}(n)
       = \outflowOfInt{\interface_1}(n) \mOp \outflowOfInt{\interface_3}(n),
     \]
     which, by cancellativity of the flow domain, implies that $\outflowOfInt{\interface_2}(n) = \outflowOfInt{\interface_3}(n)$.
     On the other hand, if $n \in \interface_1$, then
     \[
       \inflowOfInt{\interface_1}(n) = \inflowOfInt{\interface}(n) \mOp \outflowOfInt{\interface_2}(n)
       = \inflowOfInt{\interface}(n) \mOp \outflowOfInt{\interface_3}(n),
     \]
     which again, by cancellativity, implies that $\outflowOfInt{\interface_2}(n) = \outflowOfInt{\interface_3}(n)$.
  \end{itemize}
\end{proof}
}

This result forms the core of our reasoning technique; it enables us to make modifications within a chosen subgraph and, by proving preservation of its interface, know that the resulting subgraph composes with any context exactly as the original did.
Flow interfaces capture precisely the information relevant about a
flow graph, from the point of view of its context. In
\rSc{sec-expressivity} we provide additional examples of flow domains
that demonstrate the range of data structures and graph properties
that can be expressed using flows, including a notion of
\emph{universal flow} that in a sense provides a completeness result
for the expressivity of the framework.
We now turn to constructing proofs atop these new reasoning principles.

% \twfootnote{It would be good to add a short paragraph here about the completeness of flows eluding to the universal flow. We can have the details in an appendix, but I think we should mention it here.\as{ALEX: one option if we had space would be to have an ``expressiveness/examples'' subsection, in which we make both the general point and maybe list a few other flow domains and what they can capture.}}

%%% Local Variables:
%%% mode: latex
%%% TeX-master: "writeup"
%%% End:

\section{Proof Technique}\label{sec-proof-technique}

% Redo:

% start with how to specify
% two views: node-local and global interface
% Introduce predicates
% Show proof of PIP update

% Next section:
% product domains, replacement, Harris proof (as much as fits)

This section shows how to integrate flow reasoning into a standard separation logic, using the priority inheritance protocol (PIP) algorithm to illustrate our proof techniques.

Since flow graphs and flow interfaces form separation algebras, it is possible \as{in principle} to define a separation logic (SL) using these notions as \as{a custom} \emph{semantic model} (indeed, this is the proof approach taken in \cite{DBLP:journals/pacmpl/KrishnaSW18}).
By contrast, we integrate flow interfaces with a \emph{standard} separation logic without modifying its semantics.
This has the important technical advantage that our proof technique can be naturally integrated with existing separation logics and verification tools supporting SL-style reasoning.
We consider a standard \emph{sequential} SL in this section, but our technique can also be directly integrated with a concurrent SL such as RGSep (\as{as we show in} \rSc{sec-harris-proof}) or frameworks such as Iris~\cite{iris-ground-up} \as{supporting} (ghost) resources ranging over user-defined separation algebras.

\subsection{Encoding Flow-based Proofs in SL}
\label{sec-proof-technique-encoding}
%\asout{Typically, encoding a data structure invariant using flows requires placing constraints at two levels: the node-local invariant on the flow and fields of a node, and some constraint on the interface of the entire data structure (henceforth, \emph{global interface}).
%The latter is needed in order to give meaning to the flow values at each node.
%For instance, in the path-counting flow, specifying that the inflow of the global interface is $1$ at some designated node $r$ and $0$ elsewhere means that each node $n$'s flow value is the number of paths from $r$ to $n$.}
Proofs using our flow framework can employ a combination of specifications enforced at the node-level and in terms of the flow graphs and interfaces corresponding to larger heap regions such as entire data structures (henceforth, \emph{composite graphs} and \emph{composite interfaces}). At the node level, we write invariants that every node is intended to satisfy, typically relating the node's flow value to its local state (fields). For example, in the PIP, we use node-local invariants to express that a node's current priority is the maximum of the node's default priority and those in its current flow value. We typically express such specifications in terms of \emph{singleton (flow) graphs}, and their \emph{singleton interfaces}.

Specification in terms of \emph{composite} interfaces has several important purposes. One is to define custom inflows: \eg{} in the path-counting flow domain, specifying that the inflow of a composite interface is $1$ at some designated node $r$ and $0$ elsewhere enforces in any underlying flow graph that each node $n$'s flow value will be the number of paths from $r$ to $n$\footnote{Note that the analogous property cannot be captured at the node-level; when considering singleton interfaces per node in a tree rooted at $r$, \emph{every} singleton interface has an inflow of $1$.}. Composite interfaces can also be used to express that, in two states of execution, a portion of the heap ``looks the same'' with respect to composition (it has the same interface, and so can be composed with the same flow graphs), or to capture by \emph{how much} there is an observable difference in inflow or outflow; we employ this idea in the PIP proof below.

We now define an assertion syntax convenient for capturing both node-level and composite-level constraints, defined within an SL-style proof system.
%We now describe how to specify both types of constraints and reason about the effect of programs on these constraints in SL.
We assume a standard syntax and semantics of the underlying SL: see \refApp{sec-separation-logic} for more details.

\paragraph{Node Predicates}
The basic building block of our flow-based specifications is a node predicate $\nodePred(x, \fGraph)$, representing ownership of the fields of a single node $x$, as well as capturing its corresponding singleton flow graph $\fGraph$:
\begin{align*}
  \nodePred(x, \fGraph)
  & \defeq \exists \fields, \flowVar.\ x \mapsto \fields * \fGraph = (\set{x}, \lambdaFn{y}{\abstractionFn(x, \fields, y)}\ ,\ \flowVar)
    * \goodCondition(x, \fields, \flowVar(x))
\end{align*}
$\nodePred$ is implicitly parameterised by $\fields, \abstractionFn$ and $\goodCondition$; these are explained next and are typically fixed across any given flow-based proof.
The $\nodePred$ predicate expresses that we have a heap cell at location $x$ containing fields $\fields$ (a list of field-name/value mappings)\footnote{For simplicity, we assume that all fields of a flow graph node are to be handled by our flow-based technique, and that their ownership (via $\mapsto$ points-to predicates) is always carried around together; lifting these restrictions would be straightforward.}.
It also says that $\fGraph$ is a singleton flow graph with domain $\set{x}$ with some flow $\flowVar$, whose edge functions are defined by a user-defined abstraction function $\abstractionFn(x, \fields, y)$; this function allows us to define edges in terms of $x$'s field values.
Finally, the node, its fields, and its flow in this flow graph satisfy the custom predicate $\goodCondition$, used to encode node-local properties such as constraints in terms of the flow values of nodes.%\footnote{We assume here, for simplicity of presentation, that each object on the heap corresponds to a graph node, i.e., $\nodeDom = \addrs$.
%This approach can easily be extended to the case in which a graph node represents more than one heap location by replacing the $x \mapsto \fields$ assertion with an instance of a custom predicate $\hrepSpatial(x)$.}

\paragraph{Graph Predicates}
The analogous predicate for composite graphs is $\graphPred$. It carries ownership to the nodes making up potentially-unbounded graphs\sk{, using iterated separating conjunction over a \emph{set} of nodes $X$ as mentioned in \rSc{sec-introduction}:}
\[
  \graphPred(X, \fGraph) \defeq \exists \fGraphMap.\;
  \Sep_{x \in X} \nodePred(x, \fGraphMap(x))
  * \paren{\fGraphBigComp_{x \in X} \fGraphMap(x)} = \fGraph
\]
$\graphPred$ is also implicitly parameterised by $\fields, \abstractionFn$ and $\goodCondition$.
The existentially-quantified $\fGraphMap$ is a logical variable representing a \emph{function} from nodes in $X$ to corresponding singleton flow graphs. $\graphPred(X, \fGraph)$ describes a set of nodes $X$, such that each $x \in X$ is a $\nodePred$ (in particular, satisfying  $\goodCondition$), whose singleton flow graphs compose back to $\fGraph$\footnote{In specifications, we implicitly quantify at the top level over free variables such as $\fGraph$.}.
%\asfootnote{It seems intValid is superfluous in this definition now, since we pass the right composed interface in.}
%\asout{The validity predicate $\intValid{\interfaceMap(X)} \defeq \paren{\paren{\intBigComp_{x \in X} \interfaceMap(x)} \neq \bot}$ expresses the fact that the (iterated) interface composition of interfaces in $X$ exists\as{; \ie{} that these singleton interfaces have appropriate inflows and outflows to compose to \emph{some} composite interface for the domain $X$ of the current heap. In particular, }}
As well as carrying ownership of the underlying heap locations, $\graphPred$'s definition allows us to connect a node-level view of the region $X$ (each $\fGraphMap(x)$) with a composite-level view defined by $\fGraph$, on which we can impose appropriate graph-level properties such as constraints on the region's inflow.

\paragraph{Lifting to Interfaces}
Flow based proofs can often be expressed more elegantly and abstractly using predicates in terms of node and composite-level interfaces rather than flow graphs. To this end, we overload both our node and graph predicates with analogues whose second parameter is a flow interface, defined as follows:
\[
\begin{array}{rcl}
  \nodePred(x, \interface) &\; \defeq\; & \exists \fGraph.\; \nodePred(x, \fGraph) \land \interfaceFn(\fGraph) = \interface \\\
  \graphPred(X, \interface) &\; \defeq\; & \exists \fGraph.\; \graphPred(x, \fGraph) \land \interfaceFn(\fGraph) = \interface \\
\end{array}
\]
We will use these versions in the PIP proof below; interfaces capture all relevant properties for decomposition and composition of these flow graphs.

%
%\asout{Specifications. }
%
%\asout{The invariants of the data structure in question are encoded using a combination of the $\graphPred$ predicate, and another predicate $\globalInt$.
%The predicate $\globalInt$ can instead be used to constrain the \emph{composed} interface  $\interfaceMap(X)$ of the entire graph $X$, for instance expressing that it is a closed region with no outgoing edges.}
%
%\asout{For example, \todo{}.}
%
%\asout{While the \framework gives us the power to use interfaces of any size in our proofs to reason about concrete states, this choice is in some sense canonical.
%If we tie the singleton interfaces to abstract the singleton heap regions, then we can express the interface of any set of nodes as the composition of the respective singleton interfaces.
%On the other hand, if we only tied the interface of a larger region to the heap, we would lose the ability to precisely reason about modifications to single nodes when needed.}

\paragraph{Flow Lemmas}

\begin{figure}[t]
  \centering
  \begin{align}
    \graphPred(X_1 \uplus X_2, \fGraph)
    \;\; & \entails \;\; \exists \fGraph_1, \fGraph_2.\;
    \graphPred(X_1, \fGraph_1)
    * \graphPred(X_2, \fGraph_2)       \notag   \\
    & \phantom{\entails} \;\;\; * \fGraph_1 \fGraphComp \fGraph_2 = \fGraph \tag{\sc{Decomp}}\label{rule-decomp} \\
    \graphPred(X_1, \fGraph_1)
    * \graphPred(X_2, \fGraph_2)
    * \fGraph_1 \fGraphComp \fGraph_2 \neq \bot
    \;\; & \entails \;\;
           \graphPred(X_1 \uplus X_2, \fGraph_1 \fGraphComp \fGraph_2)
           \tag{\sc{Comp}}\label{rule-comp} \\
    \nodePred(x, \fGraph)
    \;\; & \equiv \;\;
           \graphPred(\set{x}, \fGraph)
           \tag{\sc{Sing}}\label{rule-singleton} \\
    \emp
    \;\; & \entails \;\;
           \graphPred(\emptyset, \fGraphEmpty)
           \tag{\sc{GrEmp}}\label{rule-gr-emp} \\
    \graphPred(X'_1, \fGraph'_1) * \graphPred(X_2, \fGraph_2)
    \land \fGraph = \fGraph_1 \fGraphComp \fGraph_2
    \;\; & \entails \;\;
           \graphPred(X'_1 \uplus X_2, \fGraph'_1 \fGraphComp \fGraph_2)
           \tag{\sc{Repl}}\label{rule-repl-equal} \\
    {} \land \interfaceFn(\fGraph_1) = \interfaceFn(\fGraph'_1)
    \;\; & \phantom{\entails} \;\;\;
           {} \land \interfaceFn(\fGraph) = \interfaceFn(\fGraph'_1 \fGraphComp \fGraph_2)
           \notag
  \end{align}
  \caption{Some useful lemmas for proving entailments between flow-based specifications.}
  \label{fig-proof-rules}
\end{figure}
We first illustrate our $\nodePred$ and $\graphPred$ assertions (capturing SL ownership of heap regions and abstracting these with flow interfaces) by identifying a number of lemmas which are generically useful in flow-based proofs. Reasoning at the level of flow interfaces is entirely in the \emph{pure} world (mathematics independent of heap-ownership and resources) with respect to the underlying SL reasoning; these lemmas are consequences of our defined assertions and the \framework{} definitions themselves.

Examples of these lemmas are shown in \rF{fig-proof-rules}. \refRule{rule-decomp} shows that we can always decompose a valid flow graph into subgraphs which are themselves flow graphs. Recomposition \refRule{rule-comp} is possible only if the interfaces of the subgraphs compose (\cf{} \defref{flow-graphs}). This rule, as well as  \refRule{rule-singleton}, and \refRule{rule-gr-emp} then all follow directly from the definition of $\graphPred$ and standard SL properties of iterated separating conjunction.
% \refRule{rule-replacement} is derived from the Replacement Theorem (\rT{thm-replacement}), and we discuss it further below.
% \refRule{rule-int-comp-in} and \refRule{rule-int-comp-out} are properties of composite interfaces that follow directly from the definition of interface composition.
The final rule \refRule{rule-repl-equal} is a direct consequence of rules \refRule{rule-comp}, \refRule{rule-decomp} and the congruence relation on flow graphs induced by their interfaces (\cf{} \rL{lem-interface-composition-congruence}). Conceptually, it expresses that after decomposing any flow graph into two parts $H_1$ and $H_2$, we can \emph{replace} $H_1$ with a new flow graph $H'_1$ with the same interface; when recomposing, the overall graph will be a flow graph with the same overall interface.

Note the connection between rules \refRule{rule-comp}/\refRule{rule-decomp} and the algebraic laws of standard inductive predicates such as $\lseg$.
For instance by combining the definition of $\graphPred$ \as{with these rules} and \refRule{rule-singleton} we can prove the following rule to fold or unfold the graph predicate:
%\begin{align}
%  \footnotesize
%  \graphPred(X \uplus \set{y}, \intBigComp_{x \in (X\uplus \set{y})} \interfaceMap(x))
%  \; \equiv \;\\
%       \nodePred(y, \interfaceMap(y))
%       * \graphPred(X, \intBigComp_{x \in X} \interfaceMap(x))
%       * \intValid{\interfaceMap(y) \intComp \intBigComp_{x \in X} \interfaceMap(x)}
%       \tag{\sc{(Un)Fold}}\label{rule-fold-unfold}
%\end{align}
\begin{align}
  \scriptsize
  \graphPred(X \uplus \set{y}, \fGraph)
  \equiv
       \exists \fGraph_y, \fGraph'. \nodePred(y, \fGraph_y)
       * \graphPred(X, \fGraph')
       * \fGraph = \fGraph_y \fGraphComp \fGraph'
       \tag{\sc{(Un)Fold}}\label{rule-fold-unfold}
\end{align}
\as{However, crucially (and unlike when using general inductive predicates \cite{ParkinsonBierman05}), this rule is symmetrical for any node $x$ in $X$; it works analogously for any desired order of decomposition of the graph, and for any data structure specified using flows.\asfootnote{Added this back as we have the space here}}
%
% This symmetry is typical of iterated separating conjunctions, but this rule captures how to interweave appropriate flow reasoning precisely. Using typical inductive predicates \cite{Parkinson},
%an advantage of using flow graph predicates is that their algebraic rules apply regardless of the data structure they describe.
%By contrast, for every inductive predicate and direction of traversal one would need to prove a new rule like \refRule{rule-fold-unfold}.

% \paragraph{Lemmas for flow interfaces}
When working with our overloaded $\nodePred$ and $\graphPred$ predicates, similar steps to those described by the above lemmas are useful. Given the these overloaded predicates, we simply apply the lemmas above to the \emph{existentially-quantified} flow-graphs in their definitions and then lift the consequence of the lemma back to the interface level using the congruence between our flow graph and interface composition notions (\rL{lem-interface-composition-congruence}).

\subsection{Proof of the PIP}

\def\adjustInflow{\delta}
\begin{figure}
  \footnotesize
  \centering
\begin{lstlisting}
// Let $\mkblue{\adjustInflow(M, q_1, q_2) \defeq M \setminus \ite{q_1 \geq 0}{\set{q_1}}{\emptyset} \cup \ite{q_2 \geq 0}{\set{q_2}}{\emptyset}}$

method update(n: Ref, from: Int, to: Int)
  requires $\mkblue{\nodePred(n, \interface_n) * \graphPred(X \setminus \set{n}, \interface') \land \interface = \interface'_n  \intComp \interface' \land \globalInt(\interface)}$
  requires $\mkblue{\interface'_n = (\set{n \goesto \adjustInflow(\inflowOfInt{\interface_n}(n), \mathit{from}, \mathit{to})}, \outflowOfInt{\interface_n}) \land \mathit{from} \neq \mathit{to}}$
  ensures $\mkblue{\graphPred(X, \interface)}$
{
  n.prios := n.prios $\setminus$ {from}
  if (to >= 0) {
    n.prios := n.prios $\cup$ {to}
  }
  from := n.curr_prio
  n.curr_prio := max(max(n.prios), n.def_prio)
  to := n.curr_prio

  if (from != to && n.next != null) {
    update(n.next, from, to)
  }
}

method acquire(p: Ref, r: Ref)
  requires $\mkblue{\graphPred(X, \interface) \land \globalInt(\interface) \land p \in X \land r \in X \land p \neq r}$
  ensures $\mkblue{\graphPred(X, \interface)}$
{
  $\annotOpt{ \exists \interface_r, \interface_p, \interface_1.\;
    \nodePred(r, \interface_r) * \nodePred(p, \interface_p)
    * \graphPred(X \setminus \set{r, p}, \interface_1)
    \land \interface = \interface_r \intComp \interface_p \intComp \interface_1
    \land \globalInt(\interface)
  }$
  if (r.next == null) {
    r.next := p; @\label{line:pip-r-next}@
    // Let $\mkpurple{q_r}$ = r.curr_prio
    $\annotOpt{
      & \exists \interface_r, \interface'_r, \interface_p, \interface_1.\;
      \nodePred(r, \interface'_r) * \nodePred(p, \interface_p)
      * \graphPred(X \setminus \set{r, p}, \interface_1)
      \land \interface = \interface_r \intComp \interface_p \intComp \interface_1 \\
      & \land \interface'_r = (\inflowOfInt{\interface_r}, \set{p \goesto \set{q_r}})
      \land \outflowOfInt{\interface_r} = \zeroFn \land \dotsc
    }$ @\label{line:pip-annot}@
    $\mkpurple{\entails}$ $\annotOpt{
      & \exists \interface_p, \interface'_p, \interface_2.\;
      \nodePred(p, \interface_p)
      * \graphPred(X \setminus \set{p}, \interface_2)
      \land \interface = \interface'_p \intComp \interface_2 \\
      & \land \interface'_p = (\set{p \goesto \delta(\inflowOfInt{\interface_p}(p), -1, q_r)}, \outflowOfInt{\interface_p})
      \land \dotsc
    }$ @\label{line:pip-annott}@
    update(p, -1, r.curr_prio)
    $\annotOpt{\graphPred(X, \interface)}$ @\label{line:pip-annottt}@
  } else {
    p.next := r; update(r, -1, p.curr_prio)
  }
}

method release(p: Ref, r: Ref)
  requires $\mkblue{\graphPred(X, \interface) \land \globalInt(\interface) \land p\in X \land r \in X \land p \neq r}$
  ensures $\mkblue{\graphPred(X, \interface)}$
{
  r.next := null; update(p, r.curr_prio, -1)
}
\end{lstlisting}
  \caption{Full PIP code and specifications, with proof sketch for \code{acquire}. \as{The comments and coloured annotations (lines \ref{line:pip-annot} to \ref{line:pip-annottt}) are used to highlight steps in the proof, and are explained in detail the text.}}
  \label{fig-pip-proof}
\end{figure}

We now have all the tools necessary to verify the priority inheritance protocol (PIP).
\rF{fig-pip-proof} gives the full algorithm with flow-based specifications; we also include some intermediate assertions to illustrate the reasoning steps for the \code{acquire} method, which we explain in more detail below.

We instantiate our framework in order to capture the PIP invariants as follows:
\begin{align*}
  \fields
  & \defeq \set{\nextField \colon y, \field{curr\_prio} \colon q, \field{def\_prio} \colon q^0, \field{prios} \colon Q} \\
  % & \\
  \abstractionFn(x, \fields, z)
  & \defeq
    \begin{cases}
      \lambdaFn{M}{\max(M \cup \{q^0\})} & \text{if } z = y \neq \nullVal \\
      \zeroFn & \text{otherwise}
    \end{cases} \\
  % & \\
  \goodCondition(x, \fields, M)
  & \defeq q^0 \geq 0 \land (\forall q' \in Q.\; q' \geq 0) \\
  & \quad \land M = Q \land q = \set{\max(Q \cup \{q^0\})} \\
  % & \\
  \globalInt(\interface)
  & \defeq
    \interface = (\set{\_ \goesto \emptyset}, \set{\_ \goesto \emptyset})
\end{align*}
Each node has the four fields listed in $\fields$.
We abstract the heap into a flow graph by letting each node have an edge to its \code{next} successor labelled by a function that passes to it the maximum incoming priority or the node's default priority: whichever is larger.
With this definition, one can see that the flow of every node will be the multiset containing \as{exactly the priorities of its predecessors}.
The node-local invariant $\goodCondition$ says that all priorities are non-negative, the flow $M$ of each node is stored in the \code{prios} field, and its current priority is the maximum of its default and incoming priorities.
Finally, the constraint $\globalInt$ on the global interface expresses that the graph is closed -- it has no inflow or outflow.

\paragraph{Specifications and Proof Outline}
Our end-to-end specifications of \code{acquire} and \code{release} guarantee that if we start with a valid flow graph (which is closed, according to $\globalInt$), we are guaranteed to return a valid flow graph with the same interface (\ie{} the graph remains closed)\asfootnote{I considered adding a remark here to the effect that keeping the identical interface allows for recomposition at call site with any other graphs, but thought it might get confusing: at least, I didn't find a good way.}. For clarity of the exposition, we focus here on how we prove that being a flow graph (with the same composite interface) is preserved; extending this specification to one which proves \eg{} that \code{acquire} adds the expected edge is straightforward in terms of standard separation logic, and we include such a specification in \refApp{sec-pip-full}.

The specification for \code{update} is more subtle. Remember that we call this function in states in which the current interfaces abstracting over the state of the whole graph's nodes \emph{cannot} compose to a flow graph; the propagation of priority information is still ongoing, and only once it completes will the nodes all satisfy their invariants and make up a flow graph. Instead, our precondition for \code{update} uses a ``fake'' interface $\interface'_n$ for the node $n$, while $n$'s current state actually matches interface $\interface_n$. The fake interface $\interface'_n$ is used to express that \emph{if} $n$ could adjust its inflow according to the propagated priority change \emph{without} changing its outflow, then it would compose back with the rest of the graph, and restore the graph's overall interface. The shorthand $\adjustInflow$ defines the required change to $n$'s inflow. In general (except when $n$'s \code{next} field is null, or $n$'s flow value is unchanged), \sk{it is not possible for $n$ to satisfy $\interface'_n$}; by updating $n$'s inflow, we will necessarily update its outflow. However, we can then construct a corresponding ``fake'' interface for the next node in the graph, reflecting the update yet to be accounted for.

\sk{
To illustrate this idea more clearly, let us consider the first \code{if}-branch in the proof of \code{acquire}.
Our intermediate proof steps are shown as purple annotations surrounded by braces.
The first step, as shown in the first line inside the method body, is to apply \refRule{rule-fold-unfold} twice (on the flow graphs represented by these predicates) and  peel off $\nodePred$ predicates for each of $r$ and $p$.
The update to $r$'s \code{next} field (line~\ref{line:pip-r-next}) causes the correct singleton interface of $r$ to change to $\interface'_r$: its outflow (previously none, since the \code{next} field was null) now propagates flow to $p$.
We summarise this state in the assertion on line~\ref{line:pip-annot} (we omit \eg{} repetition of properties from the function's precondition, focusing on the flow-related steps of the argument).
We now rewrite this state; using the definition of interface composition (\rD{def-flow-interfaces}) we deduce that although $\interface'_r$ and $\interface_p$ do not compose (since the former has outflow that the latter does not account for as inflow), the alternative ``fake'' interface $\interface'_p$ for $p$ (which artificially accounts for the missing inflow) \emph{would} do so (\cf{} line~\ref{line:pip-annott}).
Essentially, we show $\interface_r \intComp \interface_p = \interface'_r \intComp \interface'_p$, that the interface of $\set{r, p}$ would be unchanged if $p$ could somehow have interface $\interface'_p$.
Now by setting $\interface_2 = \interface'_r \intComp \interface_1$ and using algebraic properties of interfaces, we assemble the precondition expected by \code{update}.
After the call, \code{update}'s postcondition gives us the desired postcondition.
}

We focused here on the details of \code{acquire}'s proof, but very similar manipulations are required for reasoning about the recursive call in \code{update}'s implementation. The main difference there is that if the if-condition wrapping the recursive call is false then either the last-modified node has no successor (and so there is no outstanding inflow change needed), or we have \code{from = to} which implies that the ``fake'' interface is actually the same as the currently correct one.

Despite the property proved for the PIP example being a rather delicate recursive invariant over the (potentially cyclic) graph, the power of our framework enables extremely succinct specifications for the example, and proofs which require the application of relatively few generic lemmas. The integration with standard separation logic reasoning, and the parallel separation algebras provided by flow interfaces allow decomposition and recomposition to be simple proof steps. For this proof, we integrated with standard sequential separation logic, but in the next section we will show that compatibility with concurrent SL techniques is similarly straightforward.

%%% Local Variables:
%%% mode: latex
%%% TeX-master: "writeup"
%%% End:

\section{Advanced Flow Reasoning \& the Harris List}\label{sec-ea}

This section introduces some advanced \framework theory and demonstrates its use in the proof of the Harris list. \tw{We note that~\cite{DBLP:journals/pacmpl/KrishnaSW18} presented a proof of this data structure in the original flow framework. The proof given here shows that the new framework eliminates the need for the customized concurrent separation logic defined in~\cite{DBLP:journals/pacmpl/KrishnaSW18}.}
We start with a recap of Harris' algorithm \tw{adapted from~\cite{DBLP:journals/pacmpl/KrishnaSW18}}.

\subsection{The Harris List Algorithm}
\label{sec-harris-list}

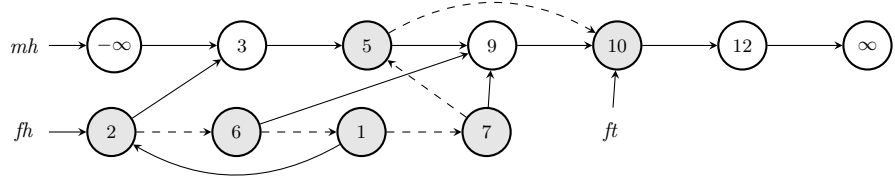
\begin{figure}[t]
  \centering
  \begin{tikzpicture}[>=stealth, font=\footnotesize, scale=0.8, every node/.style={scale=0.8}]
    % Nodes
    \node[stackVar] (hd) {$\mainListHead$};
    \node[unode, right=.5cm of hd] (inf) {$-\infty$};
    \node[unode, right= of inf] (n3) {$3$};
    \node[mnode, right= of n3] (n5) {$5$};
    \node[unode, right= of n5] (n9) {$9$};
    \node[mnode, right= of n9] (n10) {$10$};
    \node[unode, right= of n10] (n12) {$12$};
    \node[unode, right= of n12] (Inf) {$\infty$};

    %\node[stackVar] (fr) at ($(hd) + (1.5, 0)$) {$\freeListHead$};
    \node[stackVar, below=.5cm of hd] (fr) {$\freeListHead$};
    \node[mnode, right=.5cm of fr] (n2) {$2$};
    \node[mnode, right= of n2] (n6) {$6$};
    \node[mnode, right= of n6] (n1) {$1$};
    \node[mnode, right= of n1] (n7) {$7$};
    \node[stackVar, right= of n7] (ft) {$\freeListTail$};

    % Edges
    \draw[edge] (hd) to (inf);
    \draw[edge] (inf) to (n3);
    \draw[edge] (n3) to (n5);
    \draw[edge] (n5) to (n9);
    \draw[edge] (n9) to (n10);
    \draw[edge] (n10) to (n12);
    \draw[edge] (n12) to (Inf);

    \draw[edge] (fr) to (n2);
    \draw[fedge] (n2) to (n6);
    \draw[edge] (n2) to (n3);
    \draw[fedge] (n6) to (n1);
    \draw[edge] (n6) to (n9);
    \draw[fedge] (n1) to (n7);
    \draw[fedge] (n7) to (n5);
    \draw[edge] (n1) to[bend left=30] (n2);
    \draw[edge] (n7) to (n9);
    \draw[fedge] (n5) to[bend left=30] (n10);
    \draw[edge] (ft) to (n10);

  \end{tikzpicture}
  \caption[A potential state of the Harris list.]{A potential state of the Harris list with explicit memory
    management. \code{fnext} pointers are shown with dashed edges, marked nodes are shaded gray, and null pointers are omitted for clarity.}
  \label{fig-harris}
\end{figure}

The power of flow-based reasoning is exhibited in the proof of overlaid data structures such as the Harris' list, a concurrent non-blocking linked list algorithm \cite{DBLP:conf/wdag/Harris01}.
This algorithm implements a set data structure as a sorted list, and uses atomic compare-and-swap (CAS) operations to allow a high degree of parallelism.
As with the sequential linked list, Harris' algorithm inserts a new key $k$ into the list by finding nodes $k_1, k_2$ such that $k_1 < k < k_2$, setting $k$ to point to $k_2$, and using a CAS to change $k_1$ to point to $k$ only if it was still pointing to $k_2$.
However, a similar approach fails for the delete operation.
If we had consecutive nodes $k_1, k_2, k_3$ and we wanted to delete $k_2$ from the list (say by setting $k_1$ to point to $k_3$), there is no way to ensure with one CAS that $k_2$ and $k_3$ are also still adjacent (another thread could have inserted/deleted in between them).

Harris' solution is a two step deletion: first atomically mark $k_2$ as deleted (by setting a mark bit on its successor field) and then later remove it from the list using a single CAS.
After a node is marked, no thread can insert or delete to its right, hence a thread that wanted to insert $k'$ to the right of $k_2$ would first remove $k_2$ from the list and then insert $k'$ as the successor of $k_1$.

In a non-garbage-collected environment, unlinked nodes cannot be immediately freed as there may be suspended threads continuing to hold a reference to them.
A common solution is to maintain a second ``free list'' to which marked nodes are added before they are unlinked from the main list (this is the so-called drain technique).
These nodes are then labeled with a timestamp, which is used by a maintenance thread to free them when it is safe to do so.
This leads to the kind of data structure shown in \rF{fig-harris}, where each node has two pointer fields: a \code{next} field for the main list and an \code{fnext} field for the free list (shown as dashed edges).
Threads that have been suspended while holding a reference to a node that was added to the free list can simply continue traversing the \code{next} pointers to find their way back to the unmarked nodes of the main list.

Even for seemingly simple properties such as that the Harris list is memory safe and not leaking memory, the proof will rely on the following non-trivial invariants:
\begin{enumerate}[(a)]
\item \label{harris-inv-two-lists} The data structure consists of two (potentially overlapping) lists: a list on \code{next} edges beginning at $\mainListHead$ and one on \code{fnext} edges beginning at $\freeListHead$.
\item \label{harris-inv-lists-closed} The two lists are null terminated and \code{next} edges from nodes in the free list point to nodes in the free list or main list.
\item \label{harris-inv-marked} All nodes in the free list are marked.
\item \label{harris-inv-free-list-tail} $\freeListTail$ is an element in the free list.
\end{enumerate}

\subsubsection{Challenges}

\begin{figure}[t]
  \centering
  \begin{subfigure}[b]{0.15\textwidth}
    \centering
    \begin{tikzpicture}[>=stealth, scale=0.8, every node/.style={scale=0.8}, font=\footnotesize]
      \def\xsep{1}
      \def\ysep{-1.5}

      % No solution
      \node[unode] (a1) {$1$};
      \node[unode] (a2) at ($(a1) + (0, \ysep)$) {$?$};
      \node[unode] (a3) at ($(a2) + (0, \ysep)$) {$?$};
      \node[inflow] (a1-in) at ($(a1.north west) + (-.5, .5)$) {$1$};
      % A phantom node for horizontal alignment
      \node[inflow] (a1-in') at ($(a1.north east) + (.5, .5)$) {};
      % A phantom node for vertical alignment
      \node[prio] (a3-lab) at ($(a3.south) + (0, -0.25)$) {};

      \draw[edge, bend right] (a1-in) to (a1);
      \draw[edge] (a1) to (a2);
      \draw[edge, bend right] (a2) to (a3);
      \draw[edge, bend right] (a3) to (a2);
    \end{tikzpicture}
    \caption{}
    \label{fig-ea-no-solution}
  \end{subfigure}
  \begin{subfigure}[b]{0.15\textwidth}
    \centering
    \begin{tikzpicture}[>=stealth, scale=0.8, every node/.style={scale=0.8}, font=\footnotesize]
      \def\xsep{1}
      \def\ysep{-1.5}
      \def\xshift{3}

      % Many solutions
      \node[unode] (b1) {$1$};
      \node[unode] (b2) at ($(b1) + (0, \ysep)$) {$x$};
      \node[unode] (b3) at ($(b2) + (0, \ysep)$) {$x$};
      \node[inflow] (b1-in) at ($(b1.north west) + (-.5, .5)$) {$1$};
      % A phantom node for horizontal alignment
      \node[inflow] (b1-in') at ($(b1.north east) + (.5, .5)$) {};
      % A phantom node for vertical alignment
      \node[prio] (b3-lab) at ($(b3.south) + (0, -0.25)$) {};

      \draw[edge, bend right] (b1-in) to (b1);
      \draw[edge, bend right] (b2) to (b3);
      \draw[edge, bend right] (b3) to (b2);
    \end{tikzpicture}
    \caption{}
    \label{fig-ea-many-solutions}
  \end{subfigure}
  \begin{subfigure}[b]{0.6\textwidth}
    \centering
    \begin{tikzpicture}[>=stealth, scale=0.8, every node/.style={scale=0.8}, font=\footnotesize]
      \def\xsep{1.2}
      \def\ysep{-1.5}
      \def\xshift{4}
      \def\xlab{0.25}

      % Creating cycles (before)
      \node[unode] (c1) {$1$};
      \node[prio] (c1n) at ($(c1.east) + (\xlab, 0)$) {$n_1$};
      \node[unode] (c2) at ($(c1) + (2*\xsep, 0)$) {$1$};
      \node[prio] (c2n) at ($(c2.west) + (-\xlab, 0)$) {$n_2$};
      \node[unode] (c3) at ($(c1) + (0, \ysep)$) {$1$};
      \node[prio] (c3n) at ($(c3.west) + (-\xlab, 0)$) {$n_3$};
      \node[unode] (c4) at ($(c2) + (0, \ysep)$) {$1$};
      \node[prio] (c4n) at ($(c4.east) + (\xlab, 0)$) {$n_4$};
      \node[unode] (c5) at ($(c3) + (\xsep, \ysep)$) {$1$};
      \node[prio] (c5n) at ($(c5.south) + (0, -\xlab)$) {$n_5$};
      \node[inflow] (c1-in) at ($(c1.north west) + (-.5, .5)$) {$1$};

      \draw[edge, bend right] (c1-in) to (c1);
      \draw[edge] (c1) to (c3);
      \draw[edge] (c3) to (c5);
      \draw[edge] (c5) to (c4);
      \draw[edge] (c4) to (c2);

      \begin{scope}[on background layer] % Blue region around modified nodes
        \def\s{.3}
        \def\t{.1}
        \draw[draw=blue!50, rounded corners, thick, fill=blue!10]
        ($(c3n.north west) + (-\t, \s)$) -- ($(c4n.north east) + (\t, \s)$) -- ($(c4n.south east) + (\t, -\s)$) -- ($(c3n.south west) + (-\t, -\s)$) -- cycle;
      \end{scope}

      % Creating cycles (after)
      \node[unode] (d1) at ($(c1) + (\xshift+2*\xsep, 0)$) {$1$};
      \node[prio] (d1n) at ($(d1.east) + (\xlab, 0)$) {$n_1$};
      \node[unode] (d2) at ($(d1) + (2*\xsep, 0)$) {$1$};
      \node[prio] (d2n) at ($(d2.west) + (-\xlab, 0)$) {$n_2$};
      \node[unode] (d3) at ($(d1) + (0, \ysep)$) {$1$};
      \node[prio] (d3n) at ($(d3.west) + (-\xlab, 0)$) {$n_3$};
      \node[unode] (d4) at ($(d2) + (0, \ysep)$) {$1$};
      \node[prio] (d4n) at ($(d4.east) + (\xlab, 0)$) {$n_4$};
      \node[unode] (d5) at ($(d3) + (\xsep, \ysep)$) {$1$};
      \node[prio] (d5n) at ($(d5.south) + (0, -\xlab)$) {$n_5$};
      \node[inflow] (d1-in) at ($(d1.north west) + (-.5, .5)$) {$1$};
      % A phantom node for horizontal alignment
      \node[inflow] (d2-in') at ($(d2.north east) + (.5, .5)$) {};

      \draw[edge, bend right] (d1-in) to (d1);
      \draw[edge] (d1) to (d3);
      \draw[edge] (d3) to (d2);
      \draw[edge, bend right] (d4) to (d5);
      \draw[edge, bend right] (d5) to (d4);

      \begin{scope}[on background layer] % Blue region around modified nodes
        \def\s{.3}
        \def\t{.1}
        \draw[draw=blue!50, rounded corners, thick, fill=blue!10]
        ($(d3n.north west) + (-\t, \s)$) -- ($(d4n.north east) + (\t, \s)$) -- ($(d4n.south east) + (\t, -\s)$) -- ($(d3n.south west) + (-\t, -\s)$) -- cycle;
      \end{scope}

      % Squiggly arrow
      \node[phantomNode, font=\large] at ($($(c4)!0.5!(d3)$)$) {$\rightsquigarrow$};
    \end{tikzpicture}
    \caption{}
    \label{fig-ea-creating-cycles}
  \end{subfigure}
  \caption[Examples of graphs that motivate effective acyclicity.]{Examples of graphs that motivate effective acyclicity.
    All graphs use the path-counting flow domain, the flow is displayed inside each node, and the inflow is displayed as curved arrows to the top-left of nodes.
    \ref{fig-ea-no-solution} shows a graph and inflow that has no solution to \rEq{eqn-flow-equation};
    \ref{fig-ea-many-solutions} has many solutions.
    \ref{fig-ea-creating-cycles} shows a modification that preserves the interface of the modified nodes, yet goes from a graph that has a unique flow to one that has many solutions to \rEq{eqn-flow-equation}.}
  \label{fig-ea-examples}
\end{figure}
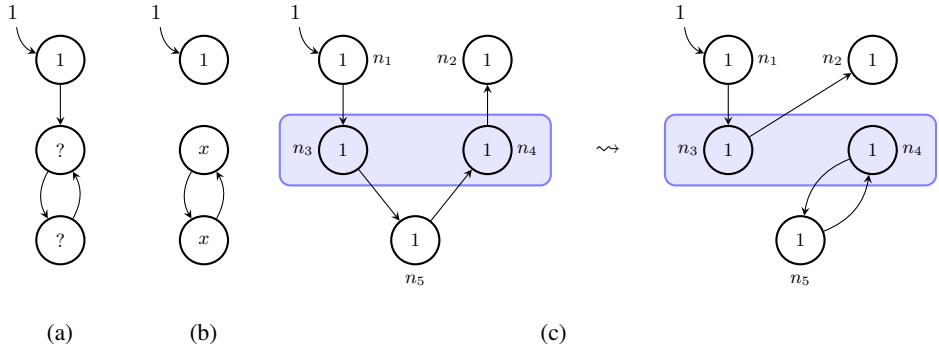

\sk{
To prove that Harris' algorithm maintains the invariants listed above we must tackle a number of challenges.
First, we must construct flow domains that allow us to describe overlaid data structures, such as the overlapping main and free lists (\rSc{sec-products-overlays}).
Second, the flow-based proofs we have seen so far center on showing that the interface of some modified region is unchanged.
However, if we consider a program that allocates and inserts a new node into a data structure (like the insert method of Harris), then the interface cannot be the same since the domain has changed (it has increased by the newly allocated node).
We must thus have a means to reason about preservation of flows by modifications that allocate new nodes (\rSc{sec-allocation-replacement}).
The third issue is that in some flow domains, there exist graphs $\graph$ and inflows $\inflow$ for which no solutions to the flow equation \rEq{eqn-flow-equation} exist.
For instance, consider the path-counting flow domain and the graph in \rF{fig-ea-no-solution}.
Since we would need to use the path-counting flow in the proof of the Harris list to encode its structural invariants, this presents a challenge (\rSc{sec-unique-fixpoints}).

We will next see how to overcome these three challenges in turn, and then apply those solution to the proof of the Harris list in \rSc{sec-harris-proof}.}

\subsection{Product Flows for Reasoning about Overlays}
\label{sec-products-overlays}

An important fact about flows is that any flow of a graph over a product of two flow domains is the product of the flows on each flow domain component.

\begin{lemma}
  \label{lem-fd-product}
  Given two flow domains $(\mDom_1, \mOp_1, \mZero_1, \edgeDom_1)$ and $(\mDom_2, \mOp_2, \mZero_2, \edgeDom_2)$, the \emph{product} domain $(\mDom_1 \times \mDom_2, \mOp, (\mZero_1, \mZero_2), \edgeDom_1 \times \edgeDom_2)$ where $(\mVar_1, \mVar_2) \mOp (\mVar'_1, \mVar'_2) \defeq (\mVar_1 \mOp_1 \mVar'_1, \mVar_2 \mOp_2 \mVar'_2)$ is a flow domain.
\end{lemma}

This lemma greatly simplifies reasoning about overlaid graph structures; we will use the product of two path-counting flows to describe a structure consisting of two overlaid lists that make up the Harris list.

\subsection{Contextual Extensions and The Replacement Theorem}
\label{sec-allocation-replacement}

In general, when modifying a flow graph $\fGraph$ to another flow graph $\fGraph'$, requiring that $\fGraph'$ satisfies \emph{precisely} the same interface $\interfaceFn(\fGraph)$ can be too strong a condition as it does not permit allocating \emph{new nodes}.
Instead, we want to allow $\interfaceFn(\fGraph')$ to differ from $\interfaceFn(\fGraph)$
in that the new interface could have a larger domain, as long as the new
nodes are fresh and edges from the new nodes do not change the outflow of the modified region.

\begin{definition}
  An interface $\interface = (\inflow, \outflow)$ is \emph{contextually extended} by $\interface' = (\inflow', \outflow')$, written $\interface \intLessEquiv \interface'$, if and only if
  \begin{enumerate*}[label=(\arabic{enumi})]
  \item $\dom(\inflow) \subseteq \dom(\inflow')$,
  \item $\forall n \in \dom(\inflow).\; \inflow(n) = \inflow'(n)$, and
  \item $\forall n' \not\in \dom(\inflow).\; \outflow(n') = \outflow'(n')$.
  \end{enumerate*}
\end{definition}

The following theorem states that contextual extension preserves composability and is itself preserved under interface composition.
\begin{theorem}[Replacement Theorem]
  \label{thm-replacement}
  If $\interface = \interface_1 \intComp \interface_2$, and $\interface_1 \intLessEquiv \interface'_1$ are all valid interfaces such that $\interface'_1 \cap \interface_2 = \emptyset$ and $\forall n \in \interface'_1 \setminus \interface_1.\; \outflowOfInt{\interface_2}(n) = \mZero$,
  then there exists a valid $\interface' = \interface'_1 \intComp \interface_2$ such that $\interface \intLessEquiv \interface'$.
\end{theorem}

\sk{
In terms of our flow predicates, this theorem gives rise to the following adaptation of the \refRule{rule-repl-equal} rule:
\begin{align}
  & \graphPred(X'_1, \fGraph'_1) * \graphPred(X_2, \fGraph_2)
  \land \fGraph = \fGraph_1 \fGraphComp \fGraph_2
  \land \interfaceFn(\fGraph_1) \intLessEquiv \interfaceFn(\fGraph'_1) \notag \\
  \entails \;\;
  & \exists \fGraph'.\;
  \graphPred(X'_1 \uplus X_2, \fGraph')
  \land \fGraph' = \fGraph'_1 \fGraphComp \fGraph_2
  \land \interfaceFn(\fGraph) \intLessEquiv \interfaceFn(\fGraph')
  \tag{\sc{Repl+}}\label{rule-repl-alloc}
\end{align}

The rule \refRule{rule-repl-alloc} is derived from the Replacement Theorem by letting
$\interface = \interfaceFn(\fGraph), \interface_1 = \interfaceFn(\fGraph_1), \interface_2 = \interfaceFn(\fGraph_2)$ and $\interface'_1 = \interfaceFn(\fGraph'_1)$.
We know $\interface_1 \intLessEquiv \interface'_1$, $\fGraph = \fGraph_1 \fGraphComp \fGraph_2$ tells us (by \rL{lem-interface-composition-congruence}) that $\interface = \interface_1 \intComp \interface_2$, and $\graphPred(X'_1, \fGraph'_1) * \graphPred(X_2, \fGraph_2)$ gives us $\interface'_1 \cap \interface_2 = \emptyset$.
The final condition of the Replacement Theorem is to prove that there is no outflow from $X_2$ to any newly allocated node in $X'_1$.
While we can use additional ghost state to prove such constraints in our proofs, if we assume that the memory allocator only allocates fresh addresses and restrict the abstraction function $\abstractionFn$ to only propagate flow along an edge $(n, n')$ if $n$ has a (non-ghost) field with a reference to $n'$ then this condition is always true.
For simplicity, and to keep the focus of this paper on the flow reasoning, we make this assumption in all subsequent proofs.
}

\subsection{Existence and Uniqueness of Flows}
\label{sec-unique-fixpoints}

We typically express
global properties of a graph $\graph = (N, \edgeFn)$ by fixing a global inflow
$\inflow \colon N \to \mDom$ and then constraining the flow of each node in
$N$ using node-local conditions. However, as we discussed at the beginning of this section, there is no general guarantee that a flow exists or is unique for a given $\inflow$ and $\graph$.
% While flow interfaces provide us with a way to prove locally that a modification to a subgraph preserves the flow in the rest of the graph, we still have to compute and reason about the fixpoint of the flow equation in the modified subgraph.
% One problem is that the modified subgraph can be lifted to a modified flow graph, and consequently a flow interface, only if there exists a solution to the flow equation in the modified region.
% \skfootnote{Should we explain a bit more why the inflow is fixed?}
% However, for many flow domains, given an inflow $\inflow$ and a graph $(N, \edgeFn)$ it is not clear that there exists any flow $\flow$ satisfying $\flowEqn(\inflow, \edgeFn, \flow)$.
% For instance, consider the path-counting flow domain from
% \rE{ex-fd-path-count} and the case where $(N, \edgeFn)$ contains cycles.
% In other domains, when there does exist a solution to the flow equation, there could be multiple solutions.
% An example of this is \todo{}.
% In these domains, we may want to talk about specific solutions to the flow equation in our specifications.
The remainder of this section presents two complementary conditions under which we can prove that
our flow fixpoint equation always has a unique solution. To this end, we say that a flow domain $(\mDom, \mOp, \mZero, \edgeDom)$ has \emph{unique flows} if for every graph $(N, \edgeFn)$ over this flow domain and inflow $\inflow \colon N \to \mDom$, there exists a unique $\flow$ that satisfies the flow equation $\flowEqn(\inflow, \edgeFn, \flow)$.
But first, we briefly recall some more monoid theory.

\techreport{
  \subsubsection{Positive Monoids and Endomorphisms}
}
We say $\mDom$ is \emph{positive} if $\mVar_1 \mOp \mVar_2 = \mZero$ implies that $\mVar_1 = \mVar_2 = \mZero$.
For a positive monoid $\mDom$, we can define a partial order $\mLeq$ on its elements as $\mVar_1 \mLeq \mVar_2$ if and only if $\exists \mVar_3.\; \mVar_1 \mOp \mVar_3 = \mVar_2$.
Positivity also implies that every $\mVar \in \mDom$ satisfies $\mZero \mLeq \mVar$.

%\skfootnote{Is it confusing that we are using $e$ for random functions and also for the edge labels?}
% $\pipe$ is left associative, so $\mVar \pipe e \pipe e' = e'(e(\mVar))$.
For $e, e' \colon \mDom \to \mDom$, we write $e \mOp e'$ for the function that maps $\mVar \in \mDom$ to $e(\mVar) \mOp e'(\mVar)$.
We lift this construction to a set of functions $E$ and write it as $\mBigOp_{e \in E}e$.

\begin{definition}
  A function $e \colon \mDom \to \mDom$ is called an \emph{endomorphism} on $\mDom$ if for every $\mVar_1, \mVar_2 \in \mDom$, $e(\mVar_1 \mOp \mVar_2) = e(\mVar_1) \mOp e(\mVar_2)$.
  We denote the set of all endomorphisms on $\mDom$ by $\morphisms(\mDom)$.
\end{definition}

Note that \as{for cancellative $\mDom$, for} every endomorphism $e \in \morphisms(\mDom)$, $e(\mZero) = \mZero$ by cancellativity.
\as{Note further} that $e \mOp e' \in \morphisms(\mDom)$ for any $e, e' \in \morphisms(\mDom)$.
Similarly, for $E \subseteq \morphisms(\mDom)$, $\mBigOp_{e \in E}e \in \morphisms(\mDom)$.
We say that a set of endomorphisms $E \subseteq \morphisms(\mDom)$ is \emph{closed} if for every $e, e' \in E$, $e \fnComp e' \in E$ and $e \mOp e' \in E$.

% \subsubsection{Edge-local Flows}
% \label{sec-mp-flows}

% A simple but useful case is when all edge functions $e \in \edgeDom$ \emph{ignore} their input; i.e.~are constant functions. We call such a flow domain \emph{edge-local}.
% % and always propagate a fixed value.
% %\asout{We can think of these fixed values as \emph{messages} that each node is passing to its neighbours.}
% In this case, the flow of every node can be computed as a direct aggregation (according to the flow domain operator $\mOp$) of the
% (constant) values its neighbours edge functions propagate; the flow equation is no-longer recursive and always has a unique solution.
% %messages it receives, and hence the flow equation always has a unique solution.

% \begin{example}
%   The flow of a PIP graph can be encoded using an edge-local flow
%   domain. This is because the PIP implementation tracks the flow explicitly as
%   part of the state of each object (the multisets stored in the \lstinline+prios+
%   field). We explain this in more depth in \rSc{sec-mp-transformation}.
% \end{example}

% \begin{lemma}
%   \label{lem-mp-flows-exist}
%   If $(\mDom, \mOp, \mZero, \edgeDom)$ is a flow domain such that for every $e \in \edgeDom$ there exists $a \in \mDom$ such that $e \equiv \lambdaFn{\mVar}{a}$, then this flow domain has unique flows.
% \end{lemma}

% \subsubsection{Shortest Path Domains}

% \ftodo{If we figure out a common idea behind this domain, add something to this section.}

\subsubsection{Nilpotent Cycles}
\label{sec-np-flows}

Let $(\mDom, \mOp, \mZero, \edgeDom)$ be a flow domain where every edge function $e \in \edgeDom$ is an endomorphism on $\mDom$.
In this case, we can show that the flow of a node $n$ is the sum of the flow as computed along \emph{each path} in the graph that ends at $n$.
Suppose we additionally know that the edge functions are defined such that their composition along any \emph{cycle} in the graph eventually becomes the identically zero function.
In this case, we need only consider finitely many paths to compute the flow of a node, which means the flow equation has a unique solution.

\techreport{Formally, such edge functions are called \emph{nilpotent endomorphisms}:}
\begin{definition}
  A closed set of endomorphisms $E \subseteq \morphisms(\mDom)$ is called \emph{nilpotent} if there exists $p > 1$ such that $e^p \equiv \mZero$ for every $e \in E$.
\end{definition}
%\ftodo{Do you really need a single bound on all elements? Isn't it enough that every element has some power that's zero?}

\begin{example}
  The edge functions of the inverse reachability domain of
  \secref{expressivity}\twfootnote{The reference to the appendix for the example is a bit problematic} are
  nilpotent endomorphisms (taking $p=2$).
\end{example}

\techreport{
Before we prove that nilpotent endomorphisms lead to unique flows, we present some useful notions and lemmas when dealing with flow domains that are endomorphisms.

\begin{lemma}
  \label{lem-flow-eqn-paths}
  If $(\mDom, \mOp, \mZero, \edgeDom)$ is a flow domain such that $\edgeDom$ is a closed set of endomorphisms, $\graph = (N, \edgeFn)$ is a graph, $\inflow \colon N \to \mDom$ is an inflow such that $\flowEqn(\inflow, \edgeFn, \flow)$, and $L \geq 1$,
  \begin{align*}
    \flow(n) &= \inflow(n)
    \mOp \mBigOp_{\substack{n_1, \dotsc, n_k \in N \\ 1 \leq k < L}}
    \inflow(n_1) \pipe \edgeFn(n_1, n_2) \dotsm \edgeFn(n_{k-1}, n_k) \pipe \edgeFn(n_k, n) \\
    & \quad \mOp \mBigOp_{n_1, \dotsc, n_L \in N}
      \flow(n_1) \pipe \edgeFn(n_1, n_2) \dotsm \edgeFn(n_{L-1}, n_L) \pipe \edgeFn(n_L, n).
  \end{align*}
\end{lemma}
}

We can now show that if all edges of a flow graph are labelled with edges from a nilpotent set of endomorphisms, then the flow equation has a unique solution:

\begin{lemma}
  \label{lem-nilpotent-unique}
  If $(\mDom, \mOp, \mZero, \edgeDom)$ is a flow domain such that $\mDom$ is a positive monoid and $\edgeDom$ is a nilpotent set of endomorphisms, then this flow domain has unique flows.
%  Given a graph $(N, \edgeFn)$ over this flow domain and inflow $\inflow \colon N \to \mDom$, there exists a unique $\flow$ that satisfies the flow equation $\flowEqn(\inflow, \edgeFn, \flow)$.
\end{lemma}

\subsubsection{Effectively Acyclic Flow Graphs}\label{sec-ea-flows}

There are some flow domains that compute flows useful in practice, but which do not guarantee either existence or uniqueness of fixpoints \emph{a priori} for all graphs.
For example, the path-counting flow from \rE{ex-fd-path-count} is one where for certain graphs, there exist no solutions to the flow equation (see \rF{fig-ea-no-solution}), and for others, there can exist more than one (in \rF{fig-ea-many-solutions}, the nodes marked with $x$ can have any path count, as long as they both have the same value).

In such cases, we explore how to restrict the class of \emph{graphs} we use in our flow-based proofs such that each graph has a unique fixpoint; the difficulty is that this restriction must be respected for composition of our graphs.
Here, we study the class of flow domains $(\mDom, \mOp, \mZero, \edgeDom)$ such that $\mDom$ is a positive monoid and $\edgeDom$ is a set of \emph{reduced} endomorphisms (defined below); in such domains we can decompose the flow computations into the various paths in the graph, and achieve unique fixpoints by restricting the kinds of cycles graphs can have.

\begin{definition}
  \label{def-ea}
  A flow graph $\fGraph = (N, \edgeFn, \flow)$ is \emph{effectively acyclic (EA)} if for every $1 \leq k$ and $n_1, \dotsc, n_k \in N$,
  \[ \flow(n_1) \pipe \edgeFn(n_1, n_2) \dotsb \edgeFn(n_{k-1}, n_k) \pipe \edgeFn(n_k, n_1) = \mZero. \]
\end{definition}
The simplest example of an effectively acyclic graph is one where the edges with non-zero edge functions form an acyclic graph.
However, our semantic condition is weaker: for example, when reasoning about two overlaid acyclic lists whose union happens to form a cycle, a product of two path-counting domains will satisfy effective acyclicity because the composition of different types of edges results in the zero function.

\begin{lemma}
  \label{lem-ea-flow-exists-unique}
  Let $(\mDom, \mOp, \mZero, \edgeDom)$ be a flow domain such that $\mDom$ is a positive monoid and $\edgeDom$ is a closed set of endomorphisms.
  Given a graph $(N, \edgeFn)$ over this flow domain and inflow $\inflow \colon N \to \mDom$, if there exists a flow graph $\fGraph = (N, \edgeFn, \flow)$ that is effectively acyclic, then $\flow$ is unique.
\end{lemma}

While the restriction to effectively acyclic flow graphs guarantees us that the flow is the unique fixpoint of the flow equation, it is not easy to show that modifications to the graph preserve EA while reasoning locally.
Even modifying a subgraph to another with the same flow interface (which we know guarantees that it will compose with any context) can inadvertently create a cycle in the larger composite graph.
For instance, consider \rF{fig-ea-creating-cycles}, that shows a modification to nodes $\set{n_3, n_4}$ (the boxed blue region).
The interface of this region is $(\set{n_3 \goesto 1, n_4 \goesto 1}, \set{n_5 \goesto 1, n_2 \goesto 1})$, and so swapping the edges of $n_3$ and $n_4$ preserves this interface.
However, the resulting graph, despite composing with the context to form a valid flow graph, is not EA (in this case, it has multiple solutions to the flow equation).
This shows that flow interfaces are not powerful enough to preserve effective acyclicity.
For a special class of endomorphisms, we show that a local property of the modified subgraph can be checked, which implies that the modified composite graph continues to be EA.

\begin{definition}
  A closed set of endomorphisms $\edgeDom \subseteq \morphisms(\mDom)$ is called \emph{reduced} if $e \fnComp e \equiv \zeroFn$ implies $e \equiv \zeroFn$ for every $e \in \edgeDom$.
\end{definition}

Note that if $\edgeDom$ is reduced, then no $e \in \edgeDom$ can be nilpotent.
In that sense, this class of instantiations is complementary to those
in \rSc{sec-np-flows}.

\begin{example}
  Examples of flow domains that fall into this class include
  positive semirings of reduced rings (with the additive monoid of
  the semiring being the aggregation monoid of the flow domain and $E$ being any
  set of functions that multiply their argument with a constant flow
  value). Note that any direct product of integral rings is a reduced
  ring. Hence, products of the path counting flow domain are a
  special case.
\end{example}

For reduced endomorphisms, it is sufficient to check that a modification preserves the flow routed between every pair of source and sink node.
This pairwise check ensures that we do not create any new cycles in any larger graph.
Before we can define an analogous relation to contextual extension, we first define a useful notion:

\begin{definition}
  The \emph{capacity} of a flow graph $\graph = (N, \edgeFn)$ is $\capacity(\graph) \colon N \times \nodeDom \to (\mDom \to \mDom)$ defined inductively as $\capacity(\graph) \defeq \capacity^{\abs{\graph}}(\graph)$, where $\capacity^0(\graph)(n, n') \defeq \deltaFn{n}{n'}$ and
  \begin{equation}
    \capacity^{i+1}(\graph)(n, n') \defeq \deltaFn{n}{n'} \mOp \mBigOp_{n'' \in \graph} \capacity^i(\graph)(n, n'') \fnComp \edgeFn(n'', n').
  \end{equation}
\end{definition}

For a flow graph $\fGraph = (N, \edgeFn, \flow)$, we write $\capacity(\fGraph)(n, n') = \capacity((N, \edgeFn))(n, n')$ for the capacity of the underlying graph.
Intuitively, $\capacity(\graph)(n, n')$ is the function that summarizes how flow is routed from any source node $n$ in $\graph$ to any other node $n'$, including those outside of $\graph$.

\techreport{
\begin{lemma}
  \label{lem-capacity-sum-paths}
  The capacity is equal to the following sum-of-paths expression:
  \[
    \capacity^i(\graph)(n, n') = \deltaFn{n}{n'}
    \mOp \mBigOp_{\substack{n_1, \dotsc, n_k \in \graph \\ 0 \leq k < i}}
    \edgeFn(n, n_1) \dotsm \edgeFn(n_k, n').
  \]
\end{lemma}
}

We now define a relation between flow graphs that constrains us to modifications that preserve EA while allowing us to allocate new nodes\techreport{\footnote{The monoid ordering used in the following definition exists because we are working with a positive monoid.}}.

\begin{definition}
  \label{def-subflow-preserving}
  A flow graph $\fGraph'$ is a \emph{subflow-preserving extension} of $\fGraph$, written $\fGraph \capacityExtendedBy \fGraph'$, if $\interfaceFn(\fGraph) \intLessEquiv \interfaceFn(\fGraph')$,
  \begin{align*}
    \forall n \in \fGraph, n' \not\in \fGraph', \mVar.\; &
      \mVar \mLeq \inflowFn(\fGraph)(n) \impl
      \mVar \pipe \capacity(\fGraph)(n, n') = \mVar \pipe \capacity(\fGraph')(n, n'), \text{ and} \\
    \forall n \in \fGraph' \setminus \fGraph, n' \not\in \fGraph', \mVar.\; &
      \mVar \mLeq \inflowFn(\fGraph')(n) \impl
      \mVar \pipe \capacity(\fGraph')(n, n') = \mZero.
\end{align*}
\end{definition}

We now show that it is sufficient to check our local condition on a modified subgraph to guarantee composition back to an effectively-acyclic composite graph:
\begin{theorem}
  \label{thm-replacement-ea}
  Let $(\mDom, \mOp, \mZero, \edgeDom)$ be a flow domain such that $\mDom$ is a positive monoid and $\edgeDom$ is a reduced set of endomorphisms.
  If $\fGraph = \fGraph_1 \fGraphComp \fGraph_2$ and $\fGraph_1 \capacityExtendedBy \fGraph'_1$ are all effectively acyclic flow graphs such that $\fGraph'_1 \cap \fGraph_2 = \emptyset$ and $\forall n \in \fGraph'_1 \setminus \fGraph_1.\; \outflowFn(\fGraph_2)(n) = \mZero$, then there exists an effectively acyclic flow graph $\fGraph' = \fGraph'_1 \fGraphComp \fGraph_2$ such that $\fGraph \capacityExtendedBy \fGraph'$.
\end{theorem}

\sk{
We define effectively acyclic versions of our flow graph predicates, $\nodePredEA(x, \fGraph)$ and $\graphPredEA(X, \fGraph)$, that take the same arguments but additionally constrain $\fGraph$ to be effectively acyclic.
The above theorem then implies the following entailment:
\begin{align}
  & \graphPredEA(X'_1, \fGraph'_1) * \graphPredEA(X_2, \fGraph_2)
  \land \fGraph = \fGraph_1 \fGraphComp \fGraph_2
  \land \fGraph_1 \capacityExtendedBy \fGraph'_1 \notag \\
  \entails \;\;
  & \exists \fGraph'.\;
  \graphPredEA(X'_1 \uplus X_2, \fGraph')
  \land \fGraph' = \fGraph'_1 \fGraphComp \fGraph_2
  \land \fGraph \capacityExtendedBy \fGraph'
  \tag{\sc{ReplEA}}\label{rule-repl-ea}
\end{align}
}

\subsection{Proof of the Harris List}
\label{sec-harris-proof}

We use the techniques seen in this section in the proof of Harris' list.
As the data structure consists of two potentially overlapping lists, we use \rL{lem-fd-product} to construct a product flow domain of two path-counting flows: one tracks the path count from the head of the main list, and one from the head of the free list.
We also work under the effectively acyclic restriction (i.e. we use the $\nodePredEA$ and $\graphPredEA$ predicates), both in order to obtain the desired interpretation of the flow as well as to ensure existence of flows in this flow domain.

We instantiate the framework using the following definitions of parameters:
\begin{align*}
  \fields
  & \defeq \set{\field{key} \colon k, \nextField \colon y, \fnextField \colon z} \\
  \abstractionFn(x, \fields, v)
  & \defeq (v = \nullVal \; ? \; \zeroFn : (v = y \land y \neq z \; ? \; \lambda_{(1, 0)} \\
  & \quad \quad \quad : (v \neq y \land y = z \; ? \; \lambda_{(0, 1)} : (v = y \land y = z \; ? \; \identityFn : \zeroFn)))) \\
  \goodCondition(x, \fields, \interface)
  & \defeq (\inflowOfInt{\interface}(x) \in \set{(1, 0), (0, 1), (1, 1)}) \land (\inflowOfInt{\interface}(x) \neq (1, 0) \impl M(y)) \\
  & \quad \land (x = \freeListTail \impl \inflowOfInt{\interface}(x) = (\_, 1)) \land (\neg M(y) \impl z = \nullVal) \\
  \globalInt(\interface)
  & \defeq \inflowOfInt{\interface}
    = \set{\mainListHead \goesto (1, 0), \freeListHead \goesto (0, 1), \_ \goesto (0, 0)}
   \land \outflowOfInt{\interface} = \set{\_ \goesto (0, 0)}
\end{align*}
Here, $\abstractionFn$ encodes the edge functions needed to compute the product of two path counting flows, the first component tracks path-counts from $\mainListHead$ on $\nextField$ edges and the second tracks path-counts from $\freeListHead$ on $\fnextField$ edges ($\lambda_{(1, 0)} \defeq \lambdaFn{(\mVar_1, \mVar_2)}{(\mVar_1, 0)}$ and $\lambda_{(0, 1)} \defeq \lambdaFn{(\mVar_1, \mVar_2)}{(0, \mVar_2)}$).
The node-local invariant $\goodCondition$ says:
the flow is one of $\set{(1, 0), (0, 1), (1, 1)}$ (meaning that the node is on one of the two lists, invariant \ref{harris-inv-two-lists});
if the flow is not $(1, 0)$ (the node is not only on the main list, i.e. it is on the free list) then the node is marked (indicated by $M(y)$, invariant \ref{harris-inv-marked});
and if the node is $\freeListTail$ then it must be on the free list (invariant \ref{harris-inv-free-list-tail}).
The constraint on the global interface, $\globalInt$, says that the inflow picks out $\mainListHead$ and $\freeListHead$ as the roots of the lists, and there is no outgoing flow (thus, all non-null edges from any node in the graph must stay within the graph, invariant \ref{harris-inv-lists-closed}).

\sk{
Since the Harris list is a concurrent algortihm, we perform the proof in rely-guarantee separation logic (RGSep)~\cite{viktor-thesis}.
Like in \rSc{sec-proof-technique}, we do not need to modify the semantics of RGSep in any way; our flow-based predicates can be defined and reasoning using our lemmas can be performed in the logic out-of-the-box.
For space reasons, the full proof can be found in \refApp{sec-harris-full}.}

%%% Local Variables:
%%% mode: latex
%%% TeX-master: "writeup"
%%% End:

%\section{Evaluation}\label{sec-evaluation}
%\input{evaluation}
\section{Related Work}\label{sec-related}
An abundance of SL variants provide complementary
mechanisms for modular reasoning about programs
(e.g.~\cite{iris-ground-up, DBLP:conf/esop/RaadVG15, DBLP:conf/pldi/SergeyNB15}). Most are parameterized by the underlying
separation algebra; our flow-based reasoning technique easily integrates with these existing logics.

\techreport{
\as{Recursive data structures are classically handled in SL using
\emph{recursive
predicates}~\cite{DBLP:conf/csl/OHearnRY01,
  DBLP:conf/lics/Reynolds02}. There is a rich line of work in automating
 such reasoning within decidable fragments (e.g.~\cite{DBLP:conf/fsttcs/BerdineCO04,DBLP:conf/atva/IosifRV14,
  DBLP:conf/tacas/KatelaanMZ19, DBLP:conf/cav/PiskacWZ13,
  DBLP:conf/nfm/EneaLSV17, DBLP:conf/vmcai/QiuW19}).
  However, recursive
  definitions are problematic for handling e.g.~graphs with cycles,
  sharing and unbounded indegree, overlaid structures and unconstrained traversals.}
}
%  While recursive predicates are
%well-suited for describing regular graph structures (e.g. lists,
%trees, etc.), they are problematic when dealing with graphs that
%involve cycles and nodes with unbounded indegree (e.g. the PIP, data
%structure overlays, algorithms operating on general graphs,
%etc.). Moreover, in order to verify that global data structure
%invariants are maintained by modifications, one needs to infer
%composition and decomposition lemmas for the recursive
%predicates. This can be done for certain decidable fragments (see
%e.g.~\cite{DBLP:conf/fsttcs/BerdineCO04,DBLP:conf/atva/IosifRV14,
%  DBLP:conf/tacas/KatelaanMZ19, DBLP:conf/cav/PiskacWZ13,
%  DBLP:conf/nfm/EneaLSV17, DBLP:conf/vmcai/QiuW19}) but generally
%relies on heuristics, in particular, if the invariants intertwine
%structural properties with properties about data. Using flows one can
%reason about structural and data properties uniformly and
%independently of each other.

The most common approach to reason about irregular graph structures in
SL is to use iterated separating
conjunction~\cite{Yang01ShorrWaite, DBLP:conf/cav/0001SS16}
and describe the graph as a set of nodes each of which satisfies some
local invariant. This approach has the advantage of being able to
naturally describe general graphs. %Moreover, modular reasoning about changes to
%the graph becomes very easy.
 However, it is hard to express
non-local properties that involve some form of fixpoint computation
over the graph structure. One approach
%A common solution to overcome this
%limitation
 is to abstract the program state as a mathematical graph
using iterated separating conjunction and then express non-local
invariants in terms of the abstract graph rather than the underlying
program state~\cite{DBLP:conf/popl/HoborV13,
DBLP:conf/pldi/SergeyNB15, DBLP:conf/aplas/RaadHVG16}. However, a
proof that a modification to the state maintains a global invariant of
the abstract graph must then often revert back to non-local \as{and manual} reasoning,
involving complex inductive arguments about paths, transitive
closure, and so on. Our technique \twout{and Viper encoding} also exploit iterated
separating conjunction for the underlying heap ownership\as{, with the key
benefit that }
%
%We build on the idea of using iterated separating conjunction in our
%proof technique and Viper encoding. The benefit of our encoding over
%prior work is that
flow interfaces exactly capture the necessary conditions on a modified subgraph in
order to compose with \emph{any} context and preserve desired non-local invariants.
%
%relies of the
%context of a graph region under modification that are needed to
%satisfy a global property of the graph. This way, we can reason
%locally about the modification without resorting to a global inductive
%proof argument.

\as{In recent work, Wang \etal{} present a Coq-mechanised proof of graph algorithms in C, based on a substantial library of graph-related lemmas, both for mathematical and heap-based graphs \cite{Wang19}. They prove rich functional properties, integrated with the VST tool. In contrast to our work, a substantial suite of lemmas and background properties are necessary, since these specialise to particular properties such as reachability. We believe that our \framework{} could be used to simplify framing lemmas in a way which remains parameteric with the property in question.
}

Proofs of a number of graph algorithms have been
mechanized in various verification tools and proof assistants,
including Trajan's SCC algorithm~\cite{DBLP:conf/itp/ChenCLMT19},
union-find~\cite{DBLP:journals/jar/ChargueraudP19}, Kruskal's minimum
spanning tree algorithm~\cite{DBLP:journals/afp/HaslbeckLB19}, and
network flow algorithms\cite{DBLP:journals/jar/LammichS19}. These
proofs generally involve non-local reasoning arguments about mathematical graphs.

The most closely related work
is~\cite{DBLP:journals/pacmpl/KrishnaSW18}, for which we already
provided a high-level comparison in \rSc{sec-introduction}. In
addition to the technical innovations made here (general proof
technique that integrates with existing SLs), the
most striking difference is in the underlying meta theory. The prior flow framework required flow domains to form a
semiring\as{; the analogue of} edge functions are restricted to multiplication
with a constant\as{, which must come from the same flow value set}.
\as{Our \framework} decouples the algebraic structure defining how
flow is \emph{aggregated} from the algebraic structure of the edge
functions. In this way, we obtain a more general framework that
applies to many more examples, and with simpler flow domains.
Strictly speaking, the \as{prior and our} framework are incomparable as
the prior did not require that flow aggregation is
cancellative. As we argue in \rSc{sec-principles},
cancellativity is, in general, necessary for local reasoning, and is critical for ensuring that the inflow of a composed graph is
uniquely determined. Due to this issue,
\cite{DBLP:journals/pacmpl/KrishnaSW18} requires proofs to reason
%reasoning
 about flow
interface \emph{equivalence classes}. This prevents the general modification of graphs with cyclic structures.
%
% complicates proofs (and their
%automation, introducing quantifier alternations), and entails strong
%restrictions on modifications of cyclic structures.\twfootnote{Should
%  we discuss difference in handling uniqueness of fixpoints? \as{That would be good too, though I'm not sure if we yet have the point about e.g. the PIP naturally establishing a non-least fixpoint.}}
\skfootnote{If we really want to beat our previous selves to death, we can also point out that our notion of interface (with a pointwise outflow) is simpler than the old notion (with a pairwise flow map).}

%When we compare to existing flows works, mention that we don't need
%node labels in flow interfaces anymore.

An alternative approach to using SL-style reasoning is to commit to
global reasoning but remain within decidable logics to enable
automation~\cite{DBLP:conf/cav/ItzhakyBINS13,
  DBLP:conf/popl/KlarlundS93, DBLP:conf/cade/WiesMK11,
  DBLP:conf/popl/MadhusudanQS12, DBLP:conf/popl/LahiriQ08}. However,
such logics are restricted to certain classes of graphs and certain
types of properties. For instance, reasoning about reachability in
unbounded graphs with two successors per node is
undecidable~\cite{DBLP:conf/csl/ImmermanRRSY04}.
\as{Recent work by Ter-Gabrielyan et al.~\cite{DBLP:journals/pacmpl/Ter-GabrielyanS19}
 shows how to deal with modular framing of \emph{pairwise reachability} specifications
 in an imperative setting. Their framing notion has parallels to our
 notion of interface composition, but allows subgraphs to \emph{change}
 the paths visible to their context. The work is specific to
 a reachability relation, and cannot express the rich variety of custom graph properties available in our technique.
 % by instantiating flow domains in our technique.
}

%%% Local Variables:
%%% mode: latex
%%% TeX-master: "writeup"
%%% End:

\section{Conclusions and Future Work}\label{sec-conclusion}
We have presented the \framework{}, enabling local modular reasoning about recursively-defined properties over general graphs. The core reasoning technique has been designed to make minimal mathematical requirements, providing great flexibility in terms of potential instantiations and applications. We identified key classes of these instantiations for which we can provide existence and uniqueness guarantees for the fixpoint properties our technique addresses and demonstrate \tw{our proof technique on several challenging examples.}
%, and showed how two of these (edge-local flow domains and effectively-acyclic graphs) can be built upon to provide automated, simple proof checking for a wide variety of challenging examples.
%
\tw{As future work, we plan to automate flow-based proofs in our new
framework using existing tools that support SL-style reasoning such as Viper~\cite{viper} and GRASShopper~\cite{DBLP:conf/tacas/PiskacWZ14}.}
% As future work, we plan to investigate the potential for an extended meta-theory supporting syncing of partial dirty regions, as an alternative proof technique to our edge-local transformation. We aim to connect our techniques to other verification tools, and to mechanise the theoretical foundations. Finally, we plan to develop a lightweight front-end tool, facilitating the application of our novel reasoning to further examples. 

%%% Local Variables:
%%% mode: latex
%%% TeX-master: "writeup"
%%% End:

%% Acknowledgments
% \begin{acks}                            %% acks environment is optional
%                                         %% contents suppressed with 'anonymous'
%   %% Commands \grantsponsor{<sponsorID>}{<name>}{<url>} and
%   %% \grantnum[<url>]{<sponsorID>}{<number>} should be used to
%   %% acknowledge financial support and will be used by metadata
%   %% extraction tools.
%   This material is based upon work supported by the
%   \grantsponsor{GS100000001}{National Science
%     Foundation}{http://dx.doi.org/10.13039/100000001} under Grant
%   No.~\grantnum{GS100000001}{nnnnnnn} and Grant
%   No.~\grantnum{GS100000001}{mmmmmmm}.  Any opinions, findings, and
%   conclusions or recommendations expressed in this material are those
%   of the author and do not necessarily reflect the views of the
%   National Science Foundation.
% \end{acks}

% references aren't counted in the page limit
\newpage
%% Bibliography
\bibliographystyle{splncs04}
\bibliography{dblp,references}

% SPLIT HERE
\newpage
\appendix

\section{Separation Logic}
\label{sec-separation-logic}

Separation logic (SL), is an extension of Hoare logic~\cite{DBLP:journals/cacm/Hoare69} that is tailored to perform modular reasoning about programs that manipulate mutable resources.
The primary application of SL has been to verify heap-based data structures, but the core of SL is an abstract separation logic (based on the logic of bunched implications (BI)~\cite{DBLP:journals/bsl/OHearnP99}) that can be instantiated to obtain various existing forms of SL by choosing an appropriate resource: any separation algebra (see \rSc{sec-preliminaries}).

\paragraph{Heaps}

Our separation logic uses standard partial heaps as its semantic model.
Let us assume we have the following fixed countably infinite sets:
$\values$, consisting of program values;
$\addrs$, consisting of memory addresses;
and $\fieldDom$, consisting of of field names.
Partial heaps are partial maps from addresses to partial maps from field-names to values:
\[ \states \defeq \setcomp{h}{h \colon \addrs \pto (\fieldDom \pto \values)} \]
It is easy to see that, under the disjoint union operator $\uplus$, and using the empty heap $\heap_{\emptySubscript}$, $(\states, \uplus, \heap_{\emptySubscript})$ forms a separation algebra.

\paragraph{Programming Language}

We consider the following simple imperative programming language:
\begin{equation*}
  \begin{array}{r l L}
    C \in \cmds \Coloneqq & \skipCommand & Empty command \\
    \mid & c & Basic command \\
    \mid & C_1;\; C_2 & Sequential composition \\
    \mid & C_1 + C_2 & Non-deterministic choice \\
    \mid & C^* & Looping \\
                          &  & \\
    c \Coloneqq & \assumeCommand(B) & Assume condition \\
    \mid & \code{x := } e & Variable assignment \\
    \mid & \code{x := } e.f  & Heap dereference \\
    \mid & e_1.f \code{ := } e_2 & Heap write \\
    \mid & \code{x := } \allocCommand() & Allocate heap cell \\
    {} & \dotsc &
  \end{array}
\end{equation*}
Here, $C$ stands for commands, $c$ for basic commands, \code{x} for program variables, $e$ for heap-independent expreesions, $f \in \fieldDom$ for field names, and $B$ for boolean expressions.
Since we are only concerned with partial correctness in this dissertation, we can define the more familiar program constructs as the following syntactic shorthands:
\begin{align*}
  \code{if}(B)\; C_1 \code{ else } C_2
  & \defeq (\assumeCommand(B);\; C_1) + (\assumeCommand(\neg B);\; C_2) \\
  \code{while}(B)\; C
  & \defeq (\assumeCommand(B);\; C)^*;\; \assumeCommand(\neg B)
\end{align*}

\paragraph{Assertions}

We assume that we start from a standard first-order logic over a signature that includes a countably infinite number of uninterpreted functions and predicates.
The only requirement on the underlying logic is that it supports additional uninterpreted sorts, functions and predicates, which can be axiomatised in the pure part of the logic\footnote{We will use this power to express all the values associated with flows and flow interfaces.}.

Let $\vars$ be an infinite set of variables (we omit sorts and type-checking from the presentation, for simplicity).
The syntax of assertions $\phi$ is given by the following:
\begin{align*}
  \phi & \Coloneqq
         P \mid \true \mid \phi \land \phi \mid \phi \impl \phi \mid \exists x.\; \phi \\
       & \quad \mid e \mapsto \set{f_1 \colon e_1, \dotsc}
         \mid \phi * \phi
         \mid \Sep_{x \in X} \phi
\end{align*}
Here, the first line consists of first order assertions $P$ (called \emph{pure} assertions in the SL world), the always valid assertion $\true$, standard boolean connectives, and existential quantification.
We can define the remaining boolean connectives and universal quantification as shorthands for the appropriate combination of these.
The second line contains the new predicates and connectives introduced by SL (so-called \emph{spatial} assertions).
The \emph{points-to} assertion $e \mapsto \set{f_1 \colon e_1, \dotsc}$ is a primitive assertion that denotes a heap cell at adderss $e$ containing fields $f_1$ with value $e_1$, etc.
The key feature of SL is the new connective $*$, or \emph{separating conjunction}, that is used to conjoin two \emph{disjoint} parts of the heap.
We use the $\Sep_{x \in X} \phi$ syntax to represent iterated separating conjunction (the bound variable $x$ ranges over a set $X$)\footnote{Most presentations of SL also include the \emph{separating implication} connective $\wand$. However, logics including $\wand$ are harder to automate and usually undecidable. By omitting $\wand$ we emphasize that we do not require it to perform flow-based reasoning}.

The semantics of the separation logic assertions are defined with respect to an interpretation of (logical and program) variables $\interp \colon \vars \pto \values$.
We write $\denotation{e}_\interp$ for the denotation of expression $e$ under interpretation $\interp$.
In particular, we have:
\begin{align*}
  \begin{array}{rl c l}
    h, \interp & \models e \mapsto \set{f_1 \colon e_1, \dotsc, f_k \colon e_k} & \iff
    & h(\denotation{e}_\interp) = \set{f_1 \goesto e_1, \dotsc, f_k \goesto e_k} \\[.4em]
    h, \interp & \models \phi_1 * \phi_2 & \iff & \exists h_1, h_2.\; (h = h_1 \uplus h_2) \land (h_1, \interp \models \phi_1) \land (h_2, \interp \models \phi_2)
    %
    % h, \interp & \models \phi_1 \magicwand \phi_2 & \iff & \forall h_1, h_2.\; (h \uplus h_1 = h_2) \land (h_1, \interp \models \phi_1) \impl (h_2, \interp \models \phi_2) \\[.4em]
    % & & \dots
  \end{array}
\end{align*}
Note that the logic presented here is \emph{garbage-collected}~\cite{DBLP:conf/aplas/CaoCA17} (also known as \emph{intuitionistic}).
Thus, the semantics of the points-to assertion $x \mapsto \set{f_1 \colon e_1, \dotsc, f_k \colon e_k}$ does not restrict the heap $h$ to only contain the address $x$, it only requires $x$ to be included in its domain.
This restriction is not essential but simplifies presentation.
% ; also, \rCh{chp-automation} considers an embedding of the flow framework in the Viper tool, whose logic also uses a garbage-collected semantics.
%The core connective of separation logic is the \emph{separating conjunction}, and a heap satisfies $\phi_1 * \phi_2$ if it can be split up into two disjoint heaps $h = h_1 \uplus h_2$ such that $h_i$ satisfies $\phi_i$ for $i \in \set{1, 2}$.

\paragraph{Operational Semantics}

We give a small-step operational semantics for our programming language.
Configurations are either $\faultConfig$ or a pair $(C, \state)$ of a command $C$ and a state $\state$ (i.e. a heap-interpretation pair).
The following rules define a reduction relation $\reduces{}{}$ between configurations:
\begin{mathpar}
  \inferH{Seq1}
  {}{(\skipCommand;\; C_2), \state \reduces{}{} C_2, \state}

  \inferH{Seq2}
  {C_1, \state \reduces{}{} C'_1, \state'}
  {(C_1;\; C_2), \state \reduces{}{} (C'_1;\; C_2), \state'}

  \inferH{Cho1}
  {}{(C_1 + C_2), \state \reduces{}{} C_1, \state}

  \inferH{Cho2}
  {}{(C_1 + C_2), \state \reduces{}{} C_2, \state}

  \inferH{Ass}
  {\denotation{B}_{\state}}
  {\assumeCommand(B), \state \reduces{}{} \skipCommand, \state}

  \inferH{Loop}
  {}{C^*, \state \reduces{}{} (\skipCommand + (C; C^*)), \state}
\end{mathpar}
While we can also give similar small-step semantics to basic commands (for instance, see~\cite{DBLP:conf/lics/Reynolds02}), it is easier to understand their axiomatic semantics, presented in the next paragraph.

Soundness of separation logic, especially the frame rule below, relies on the following locality property of the semantics of the programming language.
By defining our basic commands via an axiomatic semantics, they automatically satisfy this property, and by construction all composite commands will have the locality property.

\begin{definition}[Locality]
  \begin{enumerate}[label=(L\arabic{enumi})]
  \item If $(C, \state_1 \stateComp \state) \reduces{}{}^* \faultConfig$, then $(C, \state_1) \reduces{}{}^* \faultConfig$.
  \item If $(C, \state_1 \stateComp \state) \reduces{}{}^* (\skipCommand, \state_2)$, then either there exists $\state'_2$ such that $(C, \state_1) \reduces{}{}^* (\skipCommand, \state'_2)$ and $\state_2 = \state \stateComp \state'_2$, or $(C, \state_1) \reduces{}{}^* \faultConfig$.
  \end{enumerate}
\end{definition}

\paragraph{Proof Rules}

As with Hoare logic, programs are specified in separation logic by Hoare triples.

\begin{definition}[Hoare Triple]
  We say $\models \hoareTriple{\phi}{C}{\psi}$ if for every state $\state$ such that $\state \models \phi$ we have (1) $(C, \state) \not\reduces{}{}^* \faultConfig$, and (2) for every state $\state'$ such that $(C, \state) \reduces{}{}^* (\skipCommand, \state')$, $\state' \models \psi$.
\end{definition}

In the above definition, $\reduces{}{}^*$ is the reflexive transitive closure of the reduction relation $\reduces{}{}$.
Intuitively, the judgment $\hoareTriple{\phi}{C}{\psi}$ means that if a command $C$ is executed on a state satisfying the precondition $\phi$, then it executes without faults.
Moreover, if $C$ terminates, then the resulting state satisfies the postcondition $\psi$ (thus, this is a \emph{partial} correctness criterion).

Separation logic inherits the standard Floyd-Hoare structural proof rules, and the rule of consequence:
\begin{mathpar}
  \inferH{SL-Skip}
  {}{\proves \hoareTriple{\phi}{\skipCommand}{\phi}}

  \inferH{SL-Seq}
  {\proves \hoareTriple{\phi}{C_1}{\psi} \and \proves \hoareTriple{\psi}{C_2}{\rho}}
  {\proves \hoareTriple{\phi}{C_1; C_2}{\rho}}

  \inferH{SL-Choice}
  {\proves \hoareTriple{\phi}{C_1}{\psi} \and \proves \hoareTriple{\phi}{C_2}{\psi}}
  {\proves \hoareTriple{\phi}{C_1 + C_2}{\rho}}

  \inferH{SL-Loop}
  {\proves \hoareTriple{\phi}{C}{\phi}}
  {\proves \hoareTriple{\phi}{C^*}{\phi}}

  \inferH{SL-Conseq}
  {P' \impl \phi \and \proves \hoareTriple{\phi}{C}{\psi} \and \psi \impl Q'}
  {\proves \hoareTriple{P'}{C}{Q'}}
\end{mathpar}

The scalability of SL-based reasoning arises due to the following \emph{frame rule}:
\begin{mathpar}
  \inferH{SL-Frame}
  {\proves \hoareTriple{\phi}{C}{\psi}}
  {\proves \hoareTriple{\phi * \rho}{C}{\psi * \rho}}
\end{mathpar}
The frame rule allows one to lift a proof that a command $C$ executes safely on a state satisfying $\phi$, producing a state satisfying $\psi$ if it terminates, to the setting where an additional resource $\rho$ (the \emph{frame}) is present.
Since $C$ was safe when given only $\phi$, it does not access any resources outside $\phi$; hence, $\rho$ is untouched in the postcondition.
The soundness of the frame rule relies on the disjointness of resources enforced by the separating conjunction operator $*$\footnote{The frame rule relies on a side condition that the program variables modified by $C$ do not overlap with the free variables in $\rho$, but this condition can be omitted using the ``variables as resource'' technique~\cite{Bornat:2006:VRS:1706629.1706801}.}.

For the basic commands of the programming language, one can give \emph{small axioms}, proof rules that specify the minimum resource they need in order to execute safely.
The effect of basic commands on more complex states can be derived from these and the frame rule.
Here are some of the small axioms:
\begin{mathpar}
  \inferH{SL-Assign}
  {}{\proves \hoareTriple{\psi[x \goesto e]}{\code{x := } e}{\psi}}

  \inferH{SL-Write}
  {}{\proves \hoareTriple{e_1 \mapsto \set{f \colon \_, \dotsc}}{e_1.f \code{ := } e_2}{e_1 \mapsto \set{f \colon e_2, \dotsc}}}

  \inferH{SL-Read}
  {}{\proves \hoareTriple{e \mapsto \set{f \colon z, \dotsc} * e = y}{\code{x := } e.f}{y \mapsto \set{f \colon z, \dotsc} * x = z}}
\end{mathpar}
Note that we write $\psi[x \goesto e]$ for the assertion $\psi$ where all occurrences of $x$ are replaced with $e$, and $\_$ for an anonymous existential variable (to denote expressions we do not care about).

Together with standard axioms of first-order logic, the proof rules presented above are known to be complete~\cite{yang-thesis}\footnote{Note that Yang's completeness result depends crucially on the separating implication $\wand$ being included in the assertion language.}.
In other words, all valid Hoare triples can be derived by an appropriate combinations of these axioms.

%%% Local Variables:
%%% mode: latex
%%% TeX-master: "writeup"
%%% End:

\section{Expressivity of Flows}
\label{sec-expressivity}

We now give a few examples of flows to demonstrate the range of data structures whose properties can be expressed as local constraints on each node's flow.

We start by describing a few interesting examples of flows that
capture common graph properties. The path-counting flow defined in \rSc{sec-flows} is a very useful flow for describing the shape of common structures, e.g. lists (singly and doubly linked, cyclic), trees, and (by using product flow constructions) nested and overlaid combinations of these.
By considering products with flows for data properties, we can also describe structures such as sorted lists, binary heaps, and search trees.

The next flow is similar to the PIP flow defined in \rSc{sec-flows}
and can be used to specify the correctness of algorithms
such as Dijkstra's shortest paths algorithm.

\begin{definition}[Shortest Path Flow]
  The shortest path flow uses the flow domain $(\Nat^C, \cup,
  \emptyset, \set{\lambda_n \mid n \in C})$ of multisets over costs
  $C = \Nat \cup \{\infty\}$ where $\lambda_n(S) = \{n + \min(S)\}$.
  
  Given a flow graph $\fGraph = (N, \edgeFn, \flow)$ over this domain
  and a cost labeling function $c: N \times N \to C$ for edges, if
  $\edgeFn(n, n')$ is $\lambda_{c(n,n')}$ and $\inflowFn(\fGraph) =
    \lambdaFn{n}{\ite{n = s}{\{0\}}{\emptyset}}$ for some source node
    $s$, then $\flow(n)$ is the multiset of costs of all shortest paths from
    $s$ to $n$ via $n$'s predecessors. That is, the cost of the shortest path from
    $s$ to $n$ is $\min(\flow(n))$.
\end{definition}

The next flow can be used to reason about reachability properties in graphs.

\begin{definition}[Inverse Reachability Flow]
  Consider the flow domain  $(\Nat^{2^{\nodeDom}}, \cup, \emptyset, \edgeDom)$, consisting of the monoid of multisets of sets of nodes under multiset union and edge functions $\edgeDom$ containing $\zeroFn$ and for every $n \in \nodeDom$ the function
  \[
    \lambda_{n}(S) \defeq \set{P \goesto \ite{n \in P}{S(P \setminus \set{n})}{0}}.
  \]
  Given a flow graph $\fGraph = (N, \edgeFn, \flow)$, if $\edgeFn(n, n') = \lambda_{n}$ and $\inflowFn(\fGraph) = \lambdaFn{n}{\ite{n = r}{\set{\emptyset}}{\emptyset}}$, then $\flow(n)$ is a multiset containing, for each simple path in $\fGraph$ from $r$ to $n$, the set of all nodes occurring on that path.
\end{definition}

Finally, we demonstrate the full generality of the flow equation in
terms of its ability to capture global graph properties. To this end, we define a \emph{universal} flow that computes, at each node, sufficient information to reconstruct the entire graph.
This shows that flows are powerful enough to capture any graph property of interest.
\ftodo{The edge function definition here is hard to understand - simplify?}

\begin{definition}[Universal Flow]
  \label{def-universal-flow}
  Say we are given a set of nodes $N \subseteq \nodeDom$ and a function $\epsilon \colon N \times \nodeDom \to A$ labelling each pair of nodes from some set $A$ (for instance, to encode an unlabelled graph, $A = \set{0, 1}$ and $\epsilon(n, n')$ is $1$ iff an edge is present in the graph).
  Consider the flow domain  $(\Nat^{2^{\nodeDom \times \nodeDom \times A}}, \cup, \emptyset, \edgeDom)$, consisting of the monoid of multisets of sets of tuples $(n, n', a)$ of edges $(n, n')$ and labels $a \in A$ under multiset union and edge functions $\edgeDom$ containing $\zeroFn$ and for every $n, n' \in \nodeDom, a \in A$ the function
  \[
    \lambda_{n, n', a}(S) \defeq \set{P \goesto \ite{(n, n', a) \in P}{S(P \setminus \set{(n, n', a)})}{0}}.
  \]
  Given a flow graph $\fGraph = (N, \edgeFn, \flow)$, if $\edgeFn(n, n') = \lambda_{n, n', \epsilon(n, n')}$ and $\inflowFn(\fGraph) = \lambdaFn{n}{\set{\emptyset}}$, then $\flow(n)$ is a multiset containing, for each simple path in $\fGraph$ ending at $n$, a set of all edge-label tuples of edges occurring on that path.
\end{definition}

To see why the universal flow computes the entire graph at each node, let us look at the edge functions in more detail.
The way to think of $\edgeFn(n, n') = \lambda_{n, n', \epsilon(n, n')}$ is that it looks at each path $P'$ in the input multiset $S$ and if $P'$ does not contain the tuple $(n, n', \epsilon(n, n'))$ then it adds the tuple to $P'$ and adds the resulting path $P$ to the output multiset.
In order to convert this procedure into a multiset comprehension style definition, the formal definition above starts from each path $P$ in the output multiset and works backward (i.e. $P' = P \setminus \set{(n, n', \epsilon(n, n'))}$).

To understand the flow computation, let us start with the inflow to a node $n$, the singleton multiset containing the empty set $\set{\emptyset}$, and track its progress through a path.
For every $n'$, the edge function $\edgeFn(n, n')$ acts on $\set{\emptyset}$ and propagates the singleton multiset $\set{(n, n', \epsilon(n, n'))}$.
In this way, if we consider a sequence of (distinct) edges $(n_1,
n_2), \dotsc, (n_{k-1}, n_k)$, then this value becomes
\[\set{(n_1, n_2, \epsilon(n_1, n_2)), \dotsc, (n_{k-1}, n_k,
    \epsilon(n_{k-1}, n_k))} \enspace.\]
However, the minute we follow an edge $(n_i, n_{i+1})$ that has occurred on the path before, the edge function $\edgeFn(n_i, n_{i+1})$ will send this value to the empty multiset $\emptyset$.
Thus, the flow at each node $n$ turns out to be the multiset containing sets of edge-label tuples for each simple path in the graph\footnote{This flow domain has the property that any graph has a unique solution to the flow equation (see \rSc{sec-np-flows}).}.
Note that we label all pairs of nodes in the graph by edges of the form $\lambda_{n, n', \epsilon(n, n')}$.
This means that $\flow(n)$ will contain one set for every sequence of pairs of nodes in the graph, even those corresponding to edges that do not ``exist'' in the original graph $\epsilon$.
From this information, one can easily reconstruct all of $\epsilon$ and hence any graph property of the global graph.

The power of the universal flow to capture any graph property comes with a cost: the flow footprint of any modification is the entire global graph.
This means that we lose all powers of local reasoning, and revert to expensive global reasoning about the program.
This is to be expected, however, because the universal flow captures all details of the graph, even ones that are possibly irrelevant to the correctness of the program at hand.
The art of using flows is to carefully define a flow that captures exactly the necessary global information needed to locally prove correctness of a given program.

For example, the inverse reachability flow is a simplified version of the universal flow in that for each edge it only keeps track of the source node.
By capturing less information about the global graph, however, this flow permits more modifications: for instance, one can swap the order of two nodes in a simple path and only update the flows of the two nodes modified.
This is an example of carefully tuning the flow domain to match the modifications performed by the program.

%%% Local Variables:
%%% mode: latex
%%% TeX-master: "writeup"
%%% End:

\section{The PIP}
\label{sec-pip-full}

\sk{
In order to expose field values in the top level specification (e.g. to say that acquire results in the appropriate edge) we extend the signatures of our core predicates to allow extra custom parameters:
$\goodCondition(x, \fields, \mVar, \dotsc)$ and $\nodePred(x, \fGraph, \dotsc)$.}
For the PIP, we instantiate the framework as follows, where $\eta$ is a function from nodes to nodes storing the values of the next fields (as enforced by the last line of $\goodCondition$):
\begin{align*}
  \fields
  & \defeq \set{\nextField \colon y, \field{curr\_prio} \colon q, \field{def\_prio} \colon q^0, \field{prios} \colon Q} \\
  & \\
  \abstractionFn(x, \fields, z)
  & \defeq
    \begin{cases}
      \lambdaFn{M}{\max(\max(M), q^0)} & \text{if } z = y \neq \nullVal \\
      \zeroFn & \text{otherwise}
    \end{cases} \\
  & \\
  \goodCondition(x, \fields, M, \eta)
  & \defeq q^0 \geq 0 \land (\forall q' \in Q.\; q' \geq 0) \\
  & \quad \land M = Q \land q = \set{\max(\max(Q), q^0)} \\
  & \quad \land \eta(x) = y \land \eta(x) \neq x \\
  & \\
  \globalInt(\interface)
  & \defeq
    \interface = (\set{\_ \goesto \emptyset}, \set{\_ \goesto \emptyset})
\end{align*}
\sk{Thus, $\nodePred(x, \fGraph, \eta)$ describes a node $x$ abstracted by flow graph $\fGraph$ and whose next field is $\eta(x)$.}

\begin{lstlisting}
// Let $\mkblue{\adjustInflow(M, q_1, q_2) \defeq M \setminus \ite{q_1 \geq 0}{\set{q_1}}{\emptyset} \cup \ite{q_2 \geq 0}{\set{q_2}}{\emptyset}}$
// Let $\mkblue{\Delta(\interface, n, q_1, q_2) \defeq (\set{n \goesto \adjustInflow(\inflowOfInt{\interface}(n), q_1,q_2)}, \outflowOfInt{\interface})}$

method update(n: Ref, from: Int, to: Int)
  requires $\mkblue{\nodePred(n, \interface_n, \eta) * \graphPred(X \setminus \set{n}, \interface', \eta) \land \interface = \interface'_n  \intComp \interface' \land \globalInt(\interface)}$
  requires $\mkblue{\interface'_n = \Delta(\interface_n, n, \mathit{from}, \mathit{to}) \land \mathit{from} \neq \mathit{to}}$
  ensures $\mkblue{\graphPred(X, \interface, \eta)}$
{
  n.prios := n.prios $\setminus$ {from}
  if (to >= 0) {
    n.prios := n.prios $\cup$ {to}
  }
  from := n.curr_prio
  n.curr_prio := max(max(n.prios), n.def_prio)
  to := n.curr_prio

  if (from != to && n.next != null) { // Let n' := n.next
    $\annotOpt{
      & \exists \interface'_{n'}, \interface''_n, \interface_1.\;
      \nodePred(n, \interface''_n, \eta) * \nodePred(n', \interface_{n'}, \eta)
      * \graphPred(X \setminus \set{n, n'}, \interface_1, \eta)
      \land \interface =  \interface''_n \intComp \interface'_{n'} \intComp \interface_1 \\
      & \land \interface''_n = (\inflowOfInt{\interface'_n}, \set{n' \goesto \set{\mathit{to}}})
      \land \interface'_{n'} = \Delta(\interface_{n'}, n', \mathit{from}, \mathit{to}) \land \mathit{from} \neq \mathit{to}
    }$
    $\mkpurple{\entails}$
    $\annotOpt{
      & \exists \interface'_{n'}, \interface'. \;
      \nodePred(n', \interface_{n'}, \eta) * \graphPred(X \setminus \set{n'}, \interface', \eta)
      \land \interface' \intComp \interface'_{n'} = \interface\\
      & \land \interface'_{n'} = \Delta(\interface_{n'}, n', \mathit{from}, \mathit{to})  \land \mathit{from} \neq \mathit{to}
    }$
    update(n.next, from, to)
  }
}

method acquire(p: Ref, r: Ref)
  requires $\mkblue{\graphPred(X, \interface, \eta) \land \globalInt(\interface)}$
  requires $\mkblue{p \in X \land r \in X \land p \neq r \land \eta(p) = \nullVal}$
  ensures $\mkblue{\graphPred(X, \interface, \eta')}$
  ensures $\mkblue{\ite{\eta(r) = \nullVal}{\eta' = \eta[r \goesto p]}{\eta' = \eta[p \goesto r]}}$
{
  $\annotOpt{ \exists \interface_r, \interface_p, \interface_1.\;
    \nodePred(r, \interface_r, \eta) * \nodePred(p, \interface_p, \eta)
    * \graphPred(X \setminus \set{r, p}, \interface_1, \eta)
    \land \interface = \interface_r \intComp \interface_p \intComp \interface_1
    \land \globalInt(\interface)
  }$
  if (r.next == null) {
    // Let $\mkpurple{q_r}$ = r.curr_prio ($\geq 0$ due to $\nodePred$)
    $\annotOpt{
      & \exists \interface_r, \interface_p, \interface_1.\;
      \nodePred(r, \interface_r, \eta) * \nodePred(p, \interface_p, \eta)
      * \graphPred(X \setminus \set{r, p}, \interface_1, \eta)
      \land \interface = \interface_r \intComp \interface_p \intComp \interface_1 \\
      & \land q_r \geq 0 \land \outflowOfInt{\interface_r} = \zeroFn \land \dotsc
    }$
    r.next := p
    // Let $\eta' = \eta[r\mapsto p]$
    $\annotOpt{
      & \exists \interface_r, \interface'_r, \interface_p, \interface_1.\;
      \nodePred(r, \interface'_r, \eta') * \nodePred(p, \interface_p, \eta')
      * \graphPred(X \setminus \set{r, p}, \interface_1, \eta)
      \land \interface = \interface_r \intComp \interface_p \intComp \interface_1 \\
      & \land \interface'_r = (\inflowOfInt{\interface_r}, \set{p \goesto \set{q_r}})
      \land q_r \geq 0 \land \outflowOfInt{\interface_r} = \zeroFn \land \dotsc
    }$
    $\mkpurple{\entails}$
    $\annotOpt{
      & \exists \interface_p, \interface'_p, \interface_2.\;
      \nodePred(p, \interface_p, \eta')
      * \graphPred(X \setminus \set{p}, \interface_2, \eta')
      \land \interface = \interface'_p \intComp \interface_2 \\
      & \land \interface'_p = (\set{p \goesto \delta(\inflowOfInt{\interface_p}(p), -1, q_r)}, \outflowOfInt{\interface_p})
      \land \dotsc
    }$
    update(p, -1, r.curr_prio)
    $\annotOpt{\graphPred(X, \interface, \eta')}$
  } else {
    p.next := r
    update(r, -1, p.curr_prio)
  }
}

method release(p: Ref, r: Ref)
  requires $\mkblue{\graphPred(X, \interface, \eta) \land \globalInt(\interface)}$
  requires $\mkblue{p\in X \land r \in X \land p \neq r \land \eta(r) = p}$
  ensures $\mkblue{\graphPred(X, \interface, \eta')}$
  ensures $\mkblue{\eta' = \eta[r \mapsto \nullVal]}$
{
  r.next := null
  update(p, r.curr_prio, -1)
}
\end{lstlisting}

\section{The Harris List}
\label{sec-harris-full}

\techreport{
  \ftodo{Does this algorithm end up being the same as Michael's SPAA'02 algorithm? Also check linearizability argument - if you don't need prophecy, can it be part of a lock-free template?}
}

We perform the proof of the Harris list in rely-guarantee separation logic (RGSep)~\cite{viktor-thesis}.
RGSep is parametrised by the program states (any separation algebra), the language of assertions (a variant of separation logic), and the programming language (as long as the basic commands are \emph{local}, see \rSc{sec-separation-logic}).
Unilke~\cite{DBLP:journals/pacmpl/KrishnaSW18}, where RGSep was instantiated with a bespoke separation algebra, assertion language, and programming language with custom constructs for flows, in this paper we instantiate RGSep with the standard partial heap separation algebra and the standard separation logic assertion language from \rSc{sec-separation-logic}.
We do, however, need to switch to a concurrent programming language, so we adopt the simple imperative language used in by the original RGSep presentation~\cite{viktor-thesis}.
As we did in \rSc{sec-proof-technique}, we only need to define flow predicates $\nodePred$ and $\graphPred$ within the logic and assume the flow lemmas from \rF{fig-proof-rules} in order to perform flow-based proofs in RGSep.

For the Harris, we instantiate the framework as follows, where we extend the $\goodCondition$ and $\nodePred$ predicates to also keep track of the mark status (the $M$ predicate encodes whether a reference is marked), \nextField and \fnextField values of each node:
\begin{align*}
  \fields
  & \defeq \set{\field{key} \colon k, \nextField \colon y, \fnextField \colon z} \\
  & \\
  \abstractionFn(x, \fields, v)
  & \defeq (v = \nullVal \; ? \; \zeroFn \\
  & \quad \quad : (v = y \land y \neq z \; ? \; \lambda_{(1, 0)} \\
  & \quad \quad \quad : (v \neq y \land y = z \; ? \; \lambda_{(0, 1)} \\
  & \quad \quad \quad \quad : (v = y \land y = z \; ? \; \identityFn : \zeroFn)))) \\
  & \\
  \goodCondition(x, \fields, \interface, m, x_n, x_f)
  & \defeq (\inflowOfInt{\interface}(x) \in \set{(1, 0), (0, 1), (1, 1)}) \\
  & \quad \land (\inflowOfInt{\interface}(x) \neq (1, 0) \impl M(y)) \\
  & \quad \land (x = \freeListTail \impl \inflowOfInt{\interface}(x) = (\_, 1)) \\
  & \quad \land (\neg M(y) \impl z = \nullVal) \\
  & \quad \land (M(y) \impl m = \marked_\_) \land (\neg M(y) \impl m = \unmarked) \\
  & \quad \land x_n = y \land x_f = z \\
  & \\
  \globalInt(\interface)
  & \defeq \inflowOfInt{\interface}
    = \set{\mainListHead \goesto (1, 0), \freeListHead \goesto (0, 1), \_ \goesto (0, 0)} \\
  & \quad \land \outflowOfInt{\interface} = \set{\_ \goesto (0, 0)}
  & \\
  \nodePred(x, \interface, m, x_n, x_f)
  & \defeq x \mapsto \fields * \goodCondition(x, \fields, \interface, m, x_n, x_f) \\
  & \quad * \dom(\interface) = \set{x}
    * \forall y.\; \outflowOfInt{\interface}(y) = \abstractionFn(x, \fields, y)(\inflowOfInt{\interface}(x))
\end{align*}
Note that we extend $\goodCondition$ with extra parameters $m, x_n$, and $x_f$ that keep track of the mark status (the $M$ predicate encodes whether a reference is marked), \nextField and \fnextField values of $x$.
We also expose these parameters in the  $\nodePred$ predicate.
In our proof, we ignore these additional parameters to $\nodePred$ when we do not care about them (i.e. $\nodePred(x, \interface) = \nodePred(x, \interface, \_, \_, \_)$, etc.).

\subsection{Actions}

RGSep consists of two kinds of assertions: boxed assertions $\boxed{\phi}$ talk about the shared state among threads, and unboxed assertions $\phi$ talk about thread-local state.
An RGSep proof requires an intermediate assertion in between every two atomic modifications to the shared state, along with a stability proof that this intermediate assertion is preserved by interference of other threads.
Interference is formally specified via \emph{actions}, two-state relations that describe modifications performed by threads to the shared state.

\begin{align}
  \nodePred(l, \interface_l, \unmarked)
  & \rightsquigarrow
    \begin{aligned}
      & \nodePred(l, \interface'_l, \unmarked) * \nodePred(n, \interface_n) \\
      & \land \interface_l \intLessEquiv (\interface'_l \intComp \interface_n)
    \end{aligned}
  \tag{Insert} \\
  \nodePred(r, \interface_r, \unmarked)
  & \rightsquigarrow
  \nodePred(r, \interface_r, \marked_t)
  \tag{Mark} \\
  \nodePred(\freeListTail, \interface_f, \marked_\_, x, \nullVal) * \nodePred(r, \interface_r, \marked_t, y, \nullVal)
  & \rightsquigarrow
    \begin{aligned}
      & \nodePred(\freeListTail, \interface'_f, \marked_\_, x, r) * \nodePred(r, \interface'_r, \marked_t, y, \nullVal) \\
      & \land (\interface_f \intComp \interface_r) = (\interface'_f \intComp \interface'_r)
    \end{aligned}
  \tag{Link} \\
  \nodePred(l, \interface_l, \unmarked, r, \_) * \nodePred(r, \interface_r, \marked_\_, x, y)
  & \rightsquigarrow
    \begin{aligned}
      & \nodePred(l, \interface'_l, \unmarked, x, \_) * \nodePred(r, \interface'_r, \marked_\_, x, y) \\
      & \land (\interface_l \intComp \interface_r) = (\interface'_l \intComp \interface'_r)
    \end{aligned}
    \tag{Unlink}
\end{align}

\techreport{
\ftodo{Should we also have an action for moving $\freeListTail$?}
}

\subsection{Proof}

Since we work under the effectively acyclic restriction, we use the predicates $\nodePredEA$ and $\graphPredEA$ in our proof.
As we did with $\nodePred$ and $\graphPred$, we overload  $\nodePredEA$ and $\graphPredEA$ to talk about flow interfaces instead of flow graphs:
\[
\begin{array}{rcl}
  \nodePredEA(x, \interface) &\; \defeq\; & \exists \fGraph.\; \nodePredEA(x, \fGraph) \land \interfaceFn(\fGraph) = \interface \\\
  \graphPredEA(X, \interface) &\; \defeq\; & \exists \fGraph.\; \graphPredEA(x, \fGraph) \land \interfaceFn(\fGraph) = \interface \\
\end{array}
\]
The only place where effective acyclicity reasoning comes in is when we reason about modifications to the shared state: we then have to show that the modified flow graph is a subflow-preserving extension (so that we can use the rule \refRule{rule-repl-ea}).
We show some examples of this below with purple annotations in braces.
Finally, the proof also requires showing stabilitiy of all intermediate assertions in blue.
This is straightforward and follows the proof method used by~\cite{DBLP:journals/pacmpl/KrishnaSW18,viktor-thesis}, so we omit these details here.

\begin{lstlisting}
procedure search(k: Key) returns (l: Ref, r: Ref)
  requires $\mkblue{\boxed{\exists X, \interface.\;
      \graphPredEA(X, \interface) \land \globalInt(\interface)}}$
  ensures $\mkblue{\boxed{\exists X, \interface.\;
      \graphPredEA(X, \interface) \land \globalInt(\interface)
      \land \set{l, r} \subseteq X}}$
{
  $\annot{\boxed{\exists X, \interface_\mainListHead, \interface_1.\;
      \nodePredEA(\mainListHead, \interface_\mainListHead)
      * \graphPredEA(X \setminus \set{\mainListHead}, \interface_1)
      \land \globalInt(\interface_\mainListHead \intComp \interface_1) \land \set{l, r} \subseteq X
    }}$
  l := mh; r := mh
  var n := head.next

  while (isMarkedRef(n) || r.key < k)
  invariant $\mkblue{\boxed{\exists X, \interface.\;
      \graphPredEA(X, \interface) \land \globalInt(\interface)
      \land \set{l, r} \subseteq X \land (\neg M(n) \impl n \in X)}}$
  {
    if (isMarkedRef(n)) {  // r marked
      $\annot{\boxed{\exists X, \interface_l, \interface_1.\;
          \nodePredEA(l, \interface_l) * \graphPredEA(X \setminus \set{l}, \interface_1)
          \land \globalInt(\interface_l \intComp \interface_1)
        }}$
      if (CAS(l.next, r, unmarked(n))) {  // try to unlink r
        // Success: now l.next == n, l still unmarked
        r := unmarked(n)
      } else {
        // either something inserted in between l and r or marked l, restart
        search(k)
      }
    } else {
      l := r
      r := n
    }
    $\annot{\boxed{\exists X, \interface_r, \interface_1.\;
        \nodePredEA(r, \interface_r) * \graphPredEA(X \setminus \set{r}, \interface_1)
        \land \globalInt(\interface_r \intComp \interface_1) \land l \in X
      }}$
    n := r.next
  }
  return (l, r)
}

procedure insert(k: Key)
  requires $\mkblue{\boxed{\exists X, \interface.\;
      \graphPredEA(X, \interface) \land \globalInt(\interface)}}$
  ensures $\mkblue{\boxed{\exists X, \interface.\;
      \graphPredEA(X, \interface) \land \globalInt(\interface)}}$
{
  Node n := new Node(k, null, null)
  $\annot{\boxed{\exists X, \interface.\;
      \graphPredEA(X, \interface) \land \globalInt(\interface)}
    * \nodePredEA(n, \_, \unmarked, \nullVal, \nullVal)
  }$

  while (true)
    invariant $\mkblue{\boxed{\exists X, \interface.\;
        \graphPredEA(X, \interface) \land \globalInt(\interface)}
      * \nodePredEA(n, \_, \unmarked, \nullVal, \nullVal)
    }$
  {
    var l, r := search(k)
    $\annot{\boxed{\exists X, \interface.\;
        \graphPredEA(X, \interface) \land \globalInt(\interface)
        \land \set{l, r} \subseteq X}
      * \nodePredEA(n, \_, \unmarked, \nullVal, \nullVal)
    }$
    if (r.key == k) {
      free(n)
      return false
    }
    n.next := r
    $\annot{\boxed{\exists X, \interface_l, \interface_r, \interface_1.\;
        \nodePredEA(l, \interface_l) * \graphPredEA(X \setminus \set{l}, \interface_1)
        \land \globalInt(\interface_l \intComp \interface_1)}
      * \nodePredEA(n, \_, \unmarked, \nullVal, \nullVal)
    }$
    if (CAS(l.next, r, n)) {
      $\annotOpt{ \nodePredEA(l, \fGraph'_l) * \nodePredEA(n, \fGraph'_n)
        * \graphPredEA(X \setminus \set{l}, \fGraph_1)
        \land \fGraph_l \capacityExtendedBy (\fGraph'_l \intComp \fGraph'_n)
        \land \globalInt(\interfaceFn(\fGraph_l \intComp \fGraph_1))
      }$
      return true 
    }
  }
}

procedure delete(k: Key)
  requires $\mkblue{\boxed{\exists X, \interface.\;
      \graphPredEA(X, \interface) \land \globalInt(\interface)}}$
  ensures $\mkblue{\boxed{\exists X, \interface.\;
      \graphPredEA(X, \interface) \land \globalInt(\interface)}}$
{
  var l, r, n

  while (true)
    invariant $\mkblue{\boxed{\exists X, \interface.\;
        \graphPredEA(X, \interface) \land \globalInt(\interface)}}$
  {
    l, r := search(k)
    $\annot{\boxed{\exists X, \interface.\;
        \graphPredEA(X, \interface) \land \globalInt(\interface)
        \land \set{l, r} \subseteq X}}$
    n := r.next
    if (r.key != k) {
      return false
    }
    if (!isMarkedRef(n)) {  // r unmarked
      $\annot{\boxed{\exists X, \interface_r, \interface_1.\;
          \nodePredEA(r, \interface_r) * \graphPredEA(X \setminus \set{r}, \interface_1)
          \land \globalInt(\interface_r \intComp \interface_1)
          \land l \in X} \land \neg M(n)}$
      if (CAS(r.next, n, marked(n))) {  // mark r
        $\annotOpt{
          \nodePredEA(r, \interface_r, \marked_t, \nullVal) *  \graphPredEA(X \setminus \set{r}, \interface_1)
          \land \globalInt(\interface_r \intComp \interface_1)
          \land \set{l, r} \subseteq X
        }$
        break
      }
    }  // if r already marked, we should have returned false, so retry
  }
  $\annot{\boxed{\exists X, \interface.\;
      \nodePredEA(r, \interface_r, \marked_t, \nullVal) * \graphPredEA(X \setminus \set{r}, \interface_1)
      \land \globalInt(\interface_r \intComp \interface_1) \land l \in X \land r \neq \freeListTail}}$
  // Try to unlink r from main list
  CAS(l.next, r, n)  // If this fails, next search will unlink
  $\annot{\boxed{\exists X, \interface.\;
      \nodePredEA(r, \interface_r, \marked_t, \nullVal) * \graphPredEA(X \setminus \set{r}, \interface_1)
      \land \globalInt(\interface_r \intComp \interface_1) \land l \in X \land r \neq \freeListTail}}$
  // Link r to free list
  while (!CAS(ft.next, null, r)) {}
  $\annot{\boxed{\exists X, \interface.\;
      \nodePredEA(r, \interface_r, \marked_t, \nullVal) * \graphPredEA(X \setminus \set{r}, \interface_1)
      \land \globalInt(\interface_r \intComp \interface_1) \land l \in X \land r \neq \freeListTail
      \land \inflowOfInt{\interface_r}(r) = (\_, 1)}}$
  ft := r

  return true
}
\end{lstlisting}

\end{document}